\pdfoutput=1 
\documentclass[10pt]{article}
\usepackage{amsmath, amsthm,amsfonts,amssymb,url,color,epsfig,graphics}

\usepackage{enumerate}
\usepackage{graphicx}
\usepackage[hang,small]{caption}

\newtheorem{thm}{Theorem}[section]
\newtheorem{Prop}[thm]{Proposition}
\newtheorem{lem}[thm]{Lemma}

\newtheorem{defi}[thm]{Definition}
 \newtheorem{Thm}{Theorem}[section]
 \newtheorem{Rmk}[thm]{Remark}
 \newtheorem{Lem}[thm]{Lemma}
  \newtheorem{Def}[thm]{Definition}

\def\bS {\mathbb{S}}

 \def\N {\mathbb{N}}
\def\R {\mathbb{R}}

\def\Q {\mathbb{Q}}
\def\T {\mathbb{T}}
\def\Z {\mathbb{Z}}

\def\p{\bold{p}}

\def\cB {\mathcal{B}}

\def\cD {\mathcal{D}}

\def\cH {\mathcal{H}}

\def\eps {{\varepsilon}}

\def\indc {{\bf 1}}

\def\la {\langle}
\def\ra {\rangle}

\def\d {{\partial}}

\newcommand{\Div}{\operatorname{div}}

\newcommand{\supp}{\operatorname{supp}}

\newcommand{\Ker}{\operatorname{Ker}}
\newcommand{\ba}{\begin{aligned}}
\newcommand{\ea}{\end{aligned}}
\newcommand{\be}{\begin{equation}}
\newcommand{\ee}{\end{equation}}

\numberwithin{equation}{section}

\begin{document}


\title{Attractors for two dimensional waves with  homogeneous Hamiltonians  of degree 0}
\author{Yves  Colin de Verdi\`ere 
\footnote{Universit\'e Grenoble-Alpes 
Institut Fourier,
 Unit{\'e} mixte
 de recherche CNRS-UGA 5582,
 BP 74, 38402-Saint Martin d'H\`eres Cedex (France);
{\color{blue} {\tt yves.colin-de-verdiere@univ-grenoble-alpes.fr}}}~~ \& \\
Laure Saint-Raymond
 \footnote{Ecole Normale Sup\'erieure de Lyon,  UMPA,  UMR CNRS-ENSL 5669,  46, all\'ee d'Italie,  69364-Lyon Cedex 07
(France) {\color{blue} {\tt laure.saint-raymond@ens-lyon.fr}}}
}

\maketitle 
\begin{abstract} In domains with topography, inertial and internal  waves exhibit interesting features.
 In particular, numerical and lab experiments show that, in two dimensions, 
for generic forcing frequencies, these waves concentrate on attractors. The goal of this paper is to analyze mathematically this behavior, 
using tools from spectral theory and microlocal analysis.

\end{abstract}


\section{Physical background}

The mathematical problem which will be discussed in this paper is motivated by the physical observation that, 
in presence of topography, forcing inertial or internal waves leads to the formation of singular geometric patterns,
which accumulate most of the energy.

\subsection{Inertial and internal waves}

Inertial and internal waves are of upmost importance in oceanic flows. They describe small departures from equilibrium in an incompressible fluid submitted respectively  to the Coriolis force (due to the Earth rotation), or to the combination of density stratification and  gravity. These waves are very well described when the effect of boundaries is neglected, assuming for instance that the fluid is contained in a parallelelipedic box with zero flux condition at the boundary (or equivalently in a periodic box).
For the sake of completeness, we recall here the equations governing these waves, and their usual decomposition as a superposition of plane waves using the Fourier transform.

\medskip
\noindent
$\bullet$
Let us first consider the case of inertial waves, assuming that the density of the fluid  is homogeneous $\rho(t,x) = \rho_0$. 
The incompressibility constraint states $\Div u=0$, where   $u$ is  the bulk velocity of the fluid.

The (linearized) conservation of momentum provides
$$ \rho_0 (\d_t u + u\wedge \Omega)+ \nabla p = 0 $$
denoting by $p$ the pressure, and by $\Omega$ the rotation vector. Note that there is no convection term $u\cdot \nabla u$ here as it is expected to be negligible for small fluctuations.
Assuming for the sake of simplicity that $\rho_0 = 1$ and taking the divergence of this equation, we get
$$\Div (u\wedge \Omega) +\Delta p = 0$$
from which we deduce that the pressure is given by
$$p = (-\Delta)^{-1} \Div (u\wedge \Omega) \,.$$
We therefore  end up with the dynamical equation
\begin{equation}
\label{inertial}
\d_t u + (I+ \nabla (-\Delta)^{-1}\Div)  (u\wedge \Omega) = 0\,.
\end{equation}

If the rotation is constant (say in the vertical direction $\Omega = e_3$), we rewrite the incompressibility constraint  in Fourier variables
$$ p_h \cdot \hat u_{p,h} + p_3 \hat u_{p,3} = 0,$$
using the subscript $h$ to denote  the horizontal component,
and  take the Fourier transform of (\ref{inertial}) to get
$$ 
\d _t \hat u_{p,h}+ \begin{pmatrix} {p_1p_2 \over |p|^2} & 1- {p_1^2\over |p|^2} \\ -1+{p_2^2\over |p|^2} & -{p_1p_2 \over |p|^2} \end{pmatrix} \hat u_{p,h} = 0\,.$$
The matrix  has eigenvalues $\pm i |p_3|/|p|$. In other words, the solution to the wave equation can be obtained as a superposition of plane waves $\exp (i(p\cdot x - \omega t))$ with dispersion relation
$$ \omega = \pm {|p_3| \over |p|}\,.$$

\medskip
\noindent
$\bullet$ The model leading to internal waves is a little bit more complicated. It describes an incompressible  fluid which at equilibrium is stratified in density with stable profile $\bar \rho (x) = \bar \rho(x_3)$ with $\bar \rho '(x_3) <0$. 
Small perturbations will create both a velocity field $u$ and a fluctuation of the density $\rho = \bar \rho +\eta$.

The incompressibility constraint states as previously $\Div u = 0$. If the fluctuation around equilibrium is small, the conservation of mass
$$\d_t \rho +\Div (\rho u) =0$$
gives at leading order
$$\d_t \eta + u\cdot \nabla \bar \rho = 0\,.$$

The (linearized) conservation of momentum states
$$\bar \rho \d_t u +\eta g e_3 +\nabla p = 0$$
denoting by $g$ the gravity constant, and by $p$ the pressure. As previously, the pressure is computed thanks to the incompressibility condition
$$-\Div \left( {1\over \bar \rho} \nabla  p\right)  = g \d_3 {\eta \over \bar \rho} \,.$$
In most physical systems, the variations of $\bar \rho$ are very small compared to its average $\rho_0$, and count only for the buoyancy term.
The pressure is then given by
$$ p = (-\Delta)^{-1} (g\d_3\eta)\,.$$
The system of equations can be therefore reduced to get the Boussinesq approximation
$$
\left\{
\begin{aligned}
\d_t \eta +  u_3 \bar \rho '(x_3)  = 0,\\
  \d_t u_h  + {g \over \rho_0}\nabla_h (-\Delta)^{-1}\d_3 \eta =0 , \\
  \d_t u_3 + { g\over \rho_0} (I + \d_{33} (-\Delta)^{-1}) \eta=0\,.
\end{aligned}
\right.
$$

Assuming in addition that the stratification is affine so that $\bar \rho'$ is a constant, and taking the Fourier transform of this system,
we obtain that 
$$ 
\d _t \begin{pmatrix} \hat \eta_p \\ \hat u_{p,3}\end{pmatrix} + \begin{pmatrix} 0& \bar \rho' \\ {g\over \rho_0} \left(1- {p_3^2\over |p|^2}\right)  & 0 \end{pmatrix} \begin{pmatrix} \hat \eta_p \\ \hat u_{3,h}\end{pmatrix} = 0\,.$$
The solution can be expressed as a sum of plane waves with dispersion relation
$$\omega = \pm N {|p_h| \over |p|}\,,$$
where $N= ( - g \bar \rho' /\rho_0)^{1/2}$ denotes the Brunt-V\"ais\"al\"a frequency.

\medskip

\begin{Rmk}
In both cases, the pressure is obtained by solving the Laplace equation. In a non periodic domain, we expect the solution to depend on the boundary conditions.
The Leray projection onto divergence free vector fields defined by 
$$P = (I+ \nabla (-\Delta)^{-1}\Div)$$
is a non local (pseudo-differential) operator.
\end{Rmk}

\subsection{The effect of topography}

When the domain is not periodic (or when the rotation or the Brunt-V\"ais\"al\"a frequency are not constant), the previous analysis fails because we cannot use the Fourier transform,
and it becomes complicated to handle the Leray projection. In general, the zero-flux condition on the boundary is incompatible with any decomposition on special functions.
Actually we will see that the waves exhibit a very different behaviour.

The lab experiments conducted by physicists, especially in the groups of Maas and  Dauxois \cite{MBS,dauxois, brouzet}, consist in analyzing the response of a stratified fluid to some monochromatic forcing in different geometries.
In \cite{brouzet} for instance, the domain is a 2D trapezium and the symmetry breaking is obtained by introducing a sloping boundary. 

\begin{figure} [h] 
\centering
\includegraphics[width=5cm]{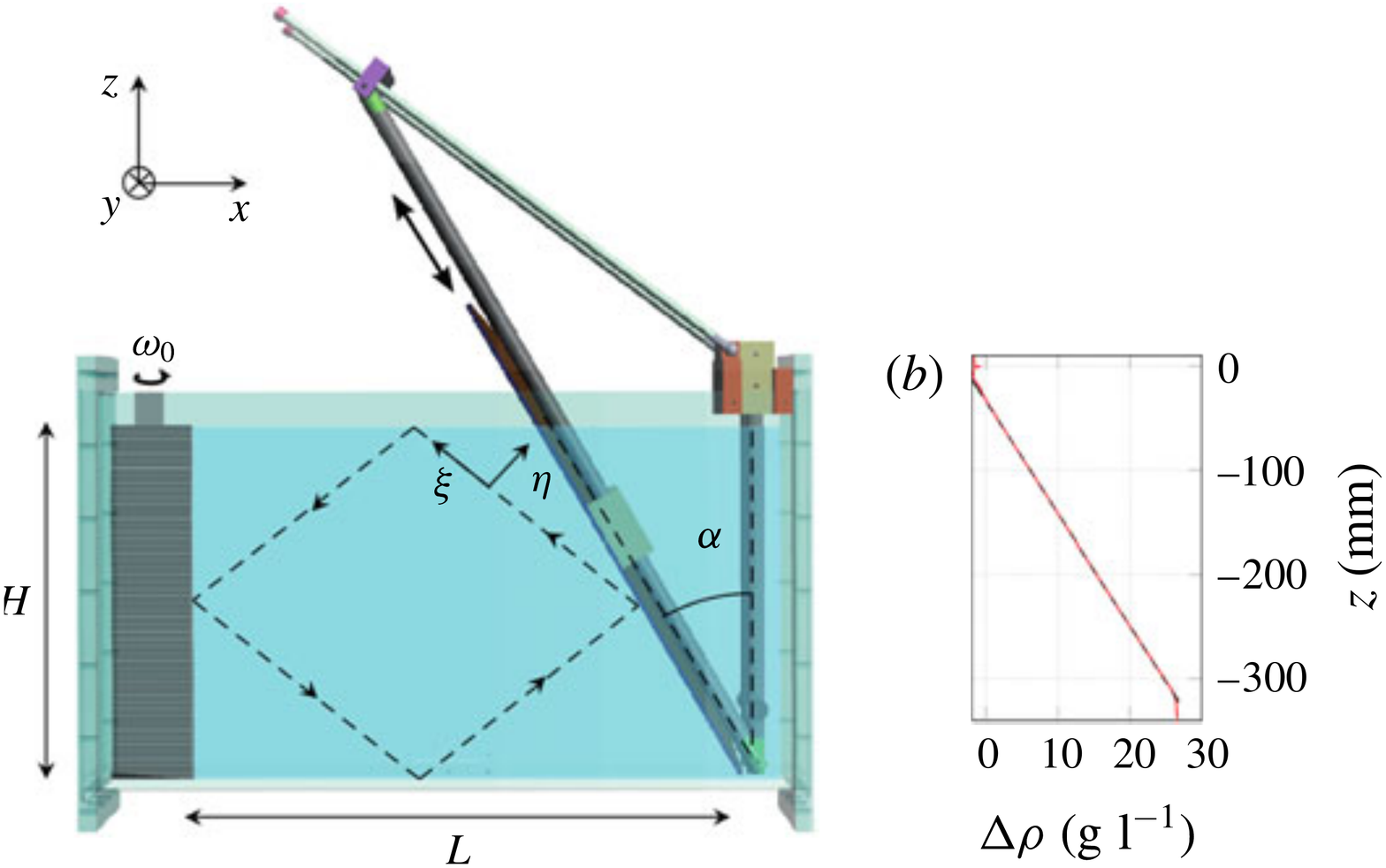} 
\caption{\small Experimental setup in \cite{brouzet}. }
\label{fig: exp-setup}
\end{figure}

 \medskip
\noindent
$\bullet$ In the absence of topography (i.e. in a rectangular box which can be seen as a torus $\T^2$), the system is completely integrable and the response to the forcing can be predicted using Fourier analysis as previously.

$\bullet$  If the forcing is non resonant, i.e. if the forcing frequency $\omega_0$ is different from $\pm |\p_h||/|p|$ for all $p$ which are excited, then the system will oscillate according to the following equation
$$ \d_t \hat v_{p,\pm} = \hat f_{p,\pm} \exp (i(\pm  \omega_p-\omega_0) t)$$
where $ \hat v_{p,\pm}$ is the amplitude of the wave $\exp (i(p\cdot x \mp \omega_p t))$.
We then obtain 
$$  \hat v_{p,\pm} =   \hat f_{p,\pm} \left( {\exp (i(\pm  \omega_p-\omega_0) t)- 1 \over i(\pm  \omega_p-\omega_0)}\right) \,.$$

Generically, if the forcing is regularly distributed over the domain (i.e. $f\in C^\infty(\T^2)$), the solution will be  regular: it is typically the case under the diophantine condition
$$\forall \omega_p \neq \omega_0, \quad  |\omega_p -\omega_0| \geq C|p|^{-\alpha}$$
which is satisfied for almost all tori.

$\bullet$ If the forcing is resonant, i.e. if the forcing frequency $\omega_0=\pm |p_h||/|p|$ for some $p$ which is excited, the fluid pumps some energy and the amplitude of the resonant mode has a linear growth in time
$$  \hat v_{p,\pm} =   \hat f_{p,\pm} t \,.$$
But, under the same generic condition, the solution is still regular in $x$ since the resonant mode has a nice oscillating structure $\exp (i(p\cdot x))$.

For the sake of completeness, we give in appendix the detailed computations leading to these regularity results.

\medskip
\noindent
$\bullet$ What is observed in the presence of topography is completely different. By PIV methods (following the displacements of small markers in the fluid), one can get a measurement of the velocity field. It happens that  this velocity field is very singular with respect to the spatial variable $x$, which is unexpected from the simple model in the periodic box. The energy concentrates on some geometric patterns, which are broken lines in the 2D trapezoidal geometry with affine stratification. 

Furthermore,
some branches of these attractors are more energetic than others, which seems to indicate that there is a focusing mechanism due to the
  reflection on the slope.
The geometry of these attractors, especially the number of branches, depends on the forcing frequency and on the slope.

\begin{figure} [h] 
\centering
\includegraphics[width=5cm]{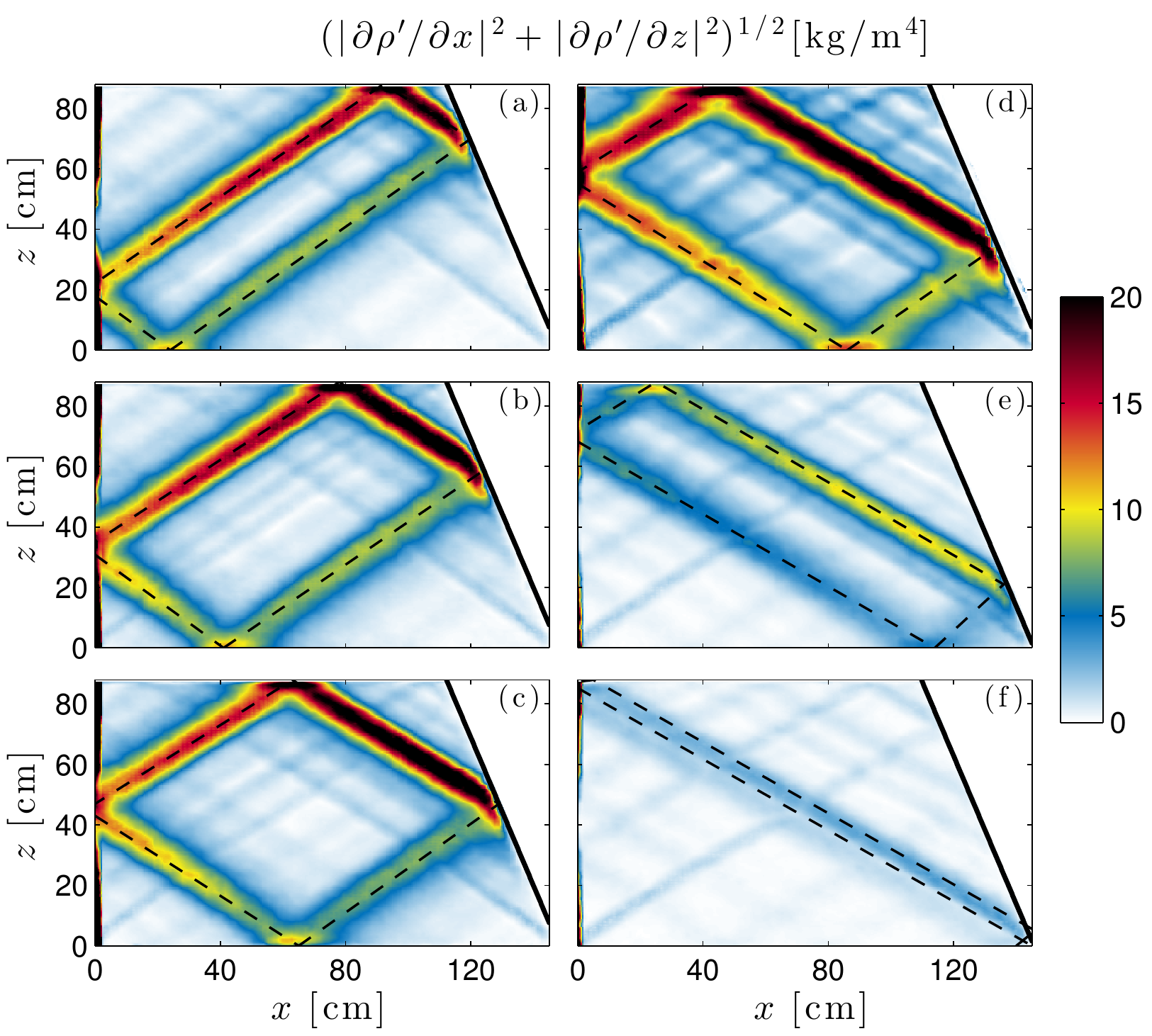} 
\caption{\small Experimental observation of attractors.\\
From the PhD thesis of C. Brouzet \cite{brouzet}, under the supervision of T. Dauxois (ENS Lyon, 2016) }
\label{fig: exp-attractor}
\end{figure}

\begin{Rmk}
In lab experiments, the effect of viscous dissipation is not negligible, which is important to explain that the energy varies also along each single branch. But we will not consider this effect here.
\end{Rmk}

\subsection{Classification of attractors}

 \medskip
\noindent
$\bullet$ The geometric pattern observed in the lab experiments coincide with the limit cycles obtained from the following geometric ray tracing
\begin{itemize}
\item start from any point $x_0$ in the domain,
\item draw the ray of direction $\phi$ such that  $\sin \phi =\omega_0/N $,
\item when the ray hits the boundary, make a reflection preserving $|\sin \phi|$,
\item continue this construction to obtain the limit cycle (or limit point).
\end{itemize}
This procedure is the counterpart of the geometric optics for electromagnetic or acoustic waves, the crucial difference here being that the frequency $\omega_0$  prescribes the direction of the propagation 
instead of the modulus of the wavelength.

 Without entering into the details of this geometric approximation, let us just explain the intuition behind it. The idea is to look at the propagation  of a wave packet, i.e. to seek  for a solution in the form
 $$ \exp \left( i{ p(t)\cdot x\over \eps}  - \frac12 (x-x(t))^2\right) $$
 which is localized both in space (around $x(t)$) and in frequency (around $p(t)$). 
 Unlike the   plane waves obtained in the first paragraph, the time frequency in this Ansatz is not quantified, and the amplitude now depends both on $t$ and $x$.
 
 \begin{Rmk}
  Of course, it is impossible to prescribe both the localization  in  $x$ and the localization in  $p$ (which is the well-known uncertainty principle in quantum physics). The main assumption behind the geometric approximation is that there is a scale separation (measured here by the factor $\eps$), and  $p$ is the Fourier variable corresponding to $x/\eps$.
  \end{Rmk}
 
 Plugging this Ansatz in the dynamical equations, we find that the propagation of the wave packet is given at leading order by the Hamiltonian equations
 $$ {dx_i\over dt } = {\d h \over \d p_i}, \quad {dp_i\over dt} = -{\d h \over \d x_i},$$
 where the Hamiltonian
is defined by the dispersion relation $h(x,p) =\omega (x,p)$.
 
 For internal waves in 2D with an affine stratification (say $N=1$), we find that $$\left\{\begin{aligned}
& \frac{dx_1}{ dt}  = \frac{p_3^2}{ |p|^3}, \quad &\frac{dp_1}{ dt} = 0,\\
 & \frac{dx_3}{ dt}  =- \frac{p_3p_1}{ |p|^3}, \quad  &\frac{dp_3}{ dt } = 0\,.
  \end{aligned}\right.
  $$
  Note that the group velocity is orthogonal to the wavenumber  $p$, and that
 its modulus is conversely proportional to $|p|$, which are  typical  features from dynamics with  Hamiltonian homogeneous of degree 0 (i.e. invariant by the homotheties $p\mapsto \lambda p$).
 This means that, on the energy level $\omega = \omega_0$, we have that the wavenumber $p$ satisfies
 $${|p_1|\over |p|} = \omega_0$$
 and the direction  of propagation, which is orthogonal to $p$, makes an angle $\phi=\pm  \arcsin \omega_0$  with respect to the horizontal.

  \medskip
\noindent
$\bullet$ Using this ray tracing method, it is possible to construct  numerically the  trajectories, and investigate systematically their  long time behavior.
 
Three scenarios appear depending on the slope $\alpha$ and on the forcing frequency $\omega_0= \sin \phi$, 
which are 
\begin{itemize}
\item convergence to a limit cycle
\item  concentration in a corner,
\item no emergence of a pattern
\end{itemize}

Note that, both  at the experimental and numerical levels, one cannot observe very long patterns because they will be generically rather dense in the domain and the accuracy of the measurements is finite.

\begin{figure} [h] 
\centering
\includegraphics[width=5.5cm]{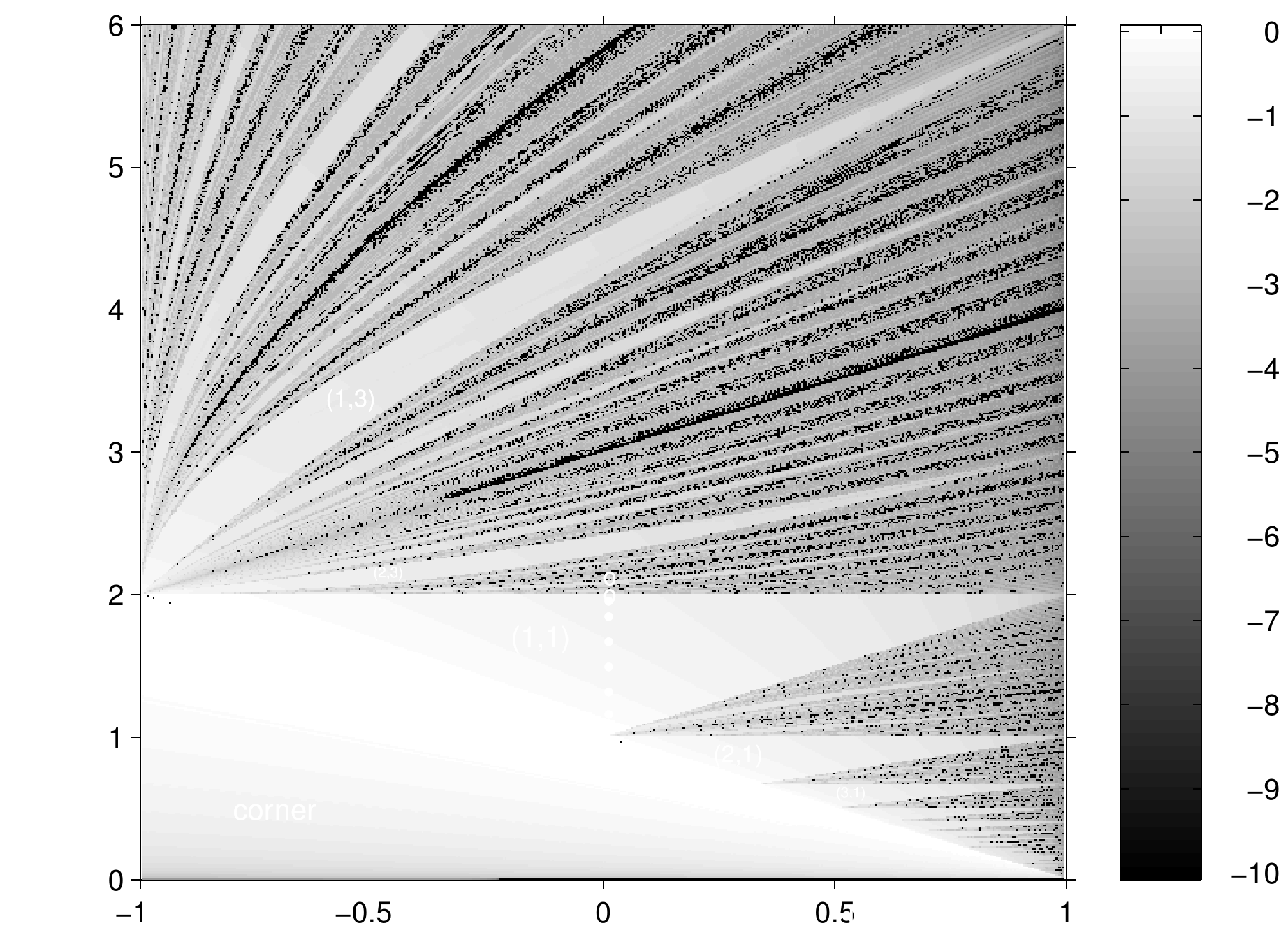} 
\caption{\small  
Greyscale: Lyapunov exponent of the trajectories.\\
White regions: attractors.\\
Black regions: no pattern emerges from the ray tracing.\\
From L. Maas, D. Benielli, J. Sommeria, F. Lam \cite{MBS}
}
\label{fig: bifurcation}
\end{figure}

At the theoretical level, the long time behavior of the Hamiltonian dynamics can be characterized by exhibiting a Poincar\'e section, and looking at the Poincar\'e return map.
If $ \phi >\alpha$, a Poincar\'e section is $x_3= 0$ and it is easy to see that all trajectories will focus on a corner.
If $ \phi<\alpha$, a Poincar\'e section is the lateral boundary $x_1= x_1^+(x_3) $.
The Poincar\'e return map $\pi_{\alpha,\phi}$ is a continuous map on  the circle $\bS^1$ (an homeomorphism) preserving the orientation. Such a map $\pi$  admits a crucial dynamical invariant $\rho(\pi) \in \T$, referred to as the rotation number, which has been introduced by Poincar\'e. It  
 can be defined for instance as  follows: for any $c,c' \in \bS^1\,$ 
$$ \rho (\pi)=  \lim_{n\to \infty} \frac1n |\{ j \in [0, n-1] \,/\, \pi^j (c')
 \in [c, \pi (c)[\} .$$
(we refer to  the first chapter of \cite{demello} for a brief presentation of the combinatorial theory of Poincar\'e).
Let us state the main properties of the  rotation number $\rho (\pi)$.
\begin{itemize}
\item When the rotation number $\rho (\pi)$ is rational, $\pi$ has periodic points (having all the same period)
 and any orbit is asymptotic to a periodic orbit.
\item When the rotation number $\rho = \rho (\pi) $ is irrational, $\pi$ is semi-conjugated to the rotation $R_\rho : \theta \mapsto \theta  + \rho $, i.e.
 there exists a monotone
 map $h$ such that $h \circ \pi= R_\rho \circ h$. Note that, in general $h$ is not a bijection: the inverse
 image of some point may be an interval.
 
  In \cite{De32}, Denjoy 
proved that, as soon as $\pi$ is smooth enough, the map $h$ is a nice continuous change of variables  and $ \pi = h^{-1} \circ R_\rho \circ h $.
\end{itemize}

For generic families of $C^2$ circle diffeomorphisms depending on one parameter,
 it has been proved by Arnold and Hermann (see Chapter 4 of \cite{demello}) that the rotation number
 is locally constant (as a function of this parameter), precisely when it is rational: this phenomenon is known
 as {\it frequency locking}.
The set of parameters for which the rotation number is irrational is a Cantor set possibly with non 
zero Lebesgue measure (devil's staircase). This provides bifurcation diagrams with Arnold's tongues
 very similar to Fig \ref{fig: bifurcation}.
Unfortunately this theory does not apply here as the Poincar\'e return map is only piecewise affine.
 However it is clear that fixed points of the (iterated) return map are stable under small perturbations
 of $\alpha$ or $\phi$. We therefore preserve the {\it band structure  of the bifurcation diagram} for rational rotation numbers.


\section{The mathematical scenario for the formation of geometric patterns}

The goal of the present paper  is to provide a mathematical description of the mechanisms leading to the formation of these geometric patterns,
 beyond the ray approximation which is not valid here since there is no scale separation.

\subsection{Quasi-resonant forcing}

The first important remark is  that the response to the forcing does not correspond to the usual picture when the wave operator has discrete eigenmodes.
A natural question is therefore to understand the response to the forcing when the wave operator has some continuous spectrum.

A very simple example which can illustrate this mechanism is the 1D Schr\"odinger equation in the whole space with a  monochromatic forcing (although it is  not really relevant from the physical point of view):
$$ \d_t \psi - i \Delta \psi = f\exp (-i\omega_0 t)\,. $$
The interest in looking at this model is that we know explicitly the spectrum $\R$ and the generalized eigenfunctions $e_\xi : x\mapsto \exp (ix\cdot \xi)$ of the Laplacian, as well as the spectral decomposition: for any $\phi \in L^2(\R)$,
$$ \phi(x)  ={1\over 2\pi} \int_\R  \hat \phi(\xi) \exp (ix\cdot \xi) d\xi\,$$
where the Fourier coefficients are defined by
$$\hat \phi(\xi)  =  \int_\R \phi(x)  \exp (-ix\cdot \xi) dx\,. $$ 

In Fourier variables, the Schr\"odinger equation can be rewritten
$$\d_t \hat \psi +i |\xi|^2 \hat \psi  = \hat f \exp (-i\omega_0 t) \,,$$
from which we deduce that 
$$ \hat \psi (t,\xi) \exp (it |\xi|^2) = -i  \hat f(\xi) { \exp (it (|\xi|^2-\omega_0)) - 1 \over |\xi|^2 - \omega_0} \,.$$

\medskip
\noindent
$\bullet$
A first way of looking at this system is to consider averages: for any $\phi \in L^2(\R)$, we define
$$ < \psi, \phi>(t)  = \int \psi(t,x) \phi^*(x)  dx =  {1\over 2\pi} \int \hat \psi(t,\xi) \bar \phi^*(\xi) d\xi \,.$$
Using Parseval's identity, we get
$$ \exp (i\omega_0 t) < \psi, \phi>(t)  =  {-i \over 2\pi} \int  \hat \phi^*(\xi)  \hat f(\xi) { 1 - \exp (- it (|\xi|^2-\omega_0))  \over |\xi|^2 - \omega_0} d\xi \,.$$
Recall  that, if  $a<\omega <b$
$$
\begin{aligned}
\int   { 1 -    \exp (-i(\lambda-\omega )   t)   \over  -\omega +\lambda }  \indc_{[a,b]} d\lambda &=  \int { 1 -    \exp (-i\mu)   \over  \mu  }  \indc_{[(a-\omega) t   ,(b-\omega)  t  ]} d\mu\\
  & = i \pi + \log {b-\omega \over \omega - a  } + O (\frac 1 t) \,,
 \end{aligned}
$$
We then expect that the average $ < \psi, \phi>(t) $ will converge to a finite value as $t \to \infty$, at least if the forcing frequency is not 0, which corresponds to a singularity of the resonant manifold.

More precisely, we have by the simple changes of variables $|\xi|^2 = \lambda = \omega_0  +{1 \over t } \mu $
$$\begin{aligned}
 2\pi < \psi, \phi>  & = -i   \exp (-i\omega_0   t) \int _0^{+\infty}  { 1 -    \exp (-i(\lambda -\omega_0 )   t)   \over  -\omega _0+\lambda  }   \hat  f(\sqrt{\lambda})  \bar \phi^*(\sqrt{\lambda}) {d\lambda\over 2\sqrt{\lambda}}\\
 &-i   \exp (-i\omega_0   t) \int _0^{+\infty}  { 1 -   \exp (-i(\lambda -\omega_0 )   t)   \over  -\omega_0 +\lambda  }   \hat  f(-\sqrt{\lambda})  \bar \phi^*(-\sqrt{\lambda}) {d\lambda\over 2\sqrt{\lambda}} \\
 & =-i    \exp (-i\omega_0   t)   \int_{-\omega_0 t} ^{+\infty}  { 1 -    \exp (-i\mu)   \over  \mu  }  F(\omega_0  +{1 \over t } \mu) d\mu \,.
  \end{aligned}$$
  with 
  $$ F(z) = \frac1{2\sqrt{z}}  ( \hat  f(\sqrt{z})  \bar \phi^*(\sqrt{z})+ \hat  f(-\sqrt{z})  \bar \phi^*(-\sqrt{z})) \,.$$
If we assume that both the forcing $f$ and the weight $\phi$ are smooth functions decaying at infinity,
  \begin{equation}
  \label{amplitude}
  \frac{2\pi}  \exp (i\omega   t)  < \psi, \phi>\to - i \la  \hbox{vp}  \left({1\over \lambda - \omega_0}\right) , F\ra + \pi F(\omega_0) \,,
  \end{equation}
  which is finite.  Note that, at this stage, we do not observe any growth in time.

\begin{Rmk}
Note that, in the integrable case, the amplitudes of the oscillating modes are special averages
$$ \psi = \sum_{n} \hat \psi_n e_n \hbox{ with } \hat \psi_n = < \psi, e_n>$$
In the resonant case, i.e. when $\omega_0 = \lambda_n$ for some $n$, the corresponding  average $< \psi, e_n>$ grows linearly in time.
\end{Rmk} 

\medskip
\noindent
$\bullet$
Physically it can be more relevant to look at the energy, or at any observable expressed as a quadratic quantity in $\psi$.
The kinetic energy for instance can be expressed in Fourier as
$$
\int |\xi|^2 |\hat \psi(t, \xi)|^2 d\xi =  \int \left| { 1 -    \exp (-i(|\xi|^2-\omega_0 )  t)   \over  -\omega_0 +|\xi|^2 } \right|^2 \ | \hat  f(\xi)  |^2 |\xi|^2 d\xi\,.$$
With the same changes of variables as previously, we get
$$
\begin{aligned}
\int |\xi|^2 |\hat \psi(t, \xi)|^2 d\xi &= \frac{1}2 \int_0^{+\infty} \left| { 1 -   \exp (-i(\lambda -\omega_0 )   t)  \over  -\omega +\lambda } \right|^2   ( |\hat  f(\sqrt{\lambda} )  |^2- |\hat f(-\sqrt{\lambda})|^2)  \sqrt{\lambda} d\lambda\\
& = \frac{ t}2 \int_{-\omega t}^{+\infty} \left| { 1 -    \exp (-i\mu )   \over  \mu } \right|^2 G(\omega _0 +{1 \over t } \mu) d\mu   \\
& =  t \int_{-\omega_0 t}^{+\infty}  { 1 -   \cos \mu    \over  \mu^2  } G(\omega_0  +{1 \over t } \mu) d\mu   
\end{aligned}$$
with 
$$G(z) =  (|\hat  f(\sqrt{z} )  |^2+ |\hat f(-\sqrt{z})|^2)  \sqrt{z}\,.$$

We then conclude that generically the energy grows linearly 
\begin{equation}
  \label{energy}
   \int |\xi|^2 |\hat \psi(t, \xi)|^2 d\xi \sim {\pi \over 2}  G(\omega_0) t \,,
   \end{equation}
provided that the support of $\hat f$ contains $\pm \sqrt{\omega_0}$.

\begin{Rmk}
Note that, in the integrable case, the energy is of the form 
$$ \sum_{n} c_n |\hat \psi_n|^2 \hbox{ with } \hat \psi_n = < \psi, e_n>$$
In the resonant case, i.e. when $\omega_0 = \lambda_n$ for some $n$, it grows quadratically in time.
\end{Rmk}

\medskip
The previous computations show that 
\begin{itemize}
\item the amplitudes of linear quantities converge to some finite values (by (\ref{amplitude}))
\item quadratic quantities grow linearly (see (\ref{energy}))
\item there is essentially no energy outside bands of width $O(1/t)$ around the forcing frequency $\omega_0$ (which corresponds to the scaling occuring in the change of variables)
\end{itemize}

This means  that the transfer of energy via quasi-resonances  is  quite  different from the classical resonance process.
The point is that the generalized eigenfunctions corresponding to the continuous part of the spectrum have infinite energy, it is therefore impossible 
that they emerge in finite time even though the forcing is concentrated on one frequency.

We will rephrase these properties in a more abstract setting (adapted to our study) in Section 4.

\subsection{Characterization of the continuous spectrum}

We therefore expect this quasi-resonant forcing to be  the right mechanism to explain the formation of the geometric patterns provided that 
\begin{itemize}
\item
one can prove that the wave operator has a continuous spectrum, 
\item   the generalized eigenfunctions exhibit some singularities on the geometric patterns.
\end{itemize}

Note that similar  spectral problems had been  investigated a long time ago by F. John in \cite{john} for the 2D wave equation with Dirichlet conditions, and by J. Ralston in \cite{ralston}  for inertial waves in some specific 2D geometries.
 We will  develop here a more systematic method: the key  idea  will be  to use microlocal analysis, i.e.
 the connection between spectral theory and classical dynamics.
\bigskip

Let us first recall that a differential operator of order $m$  can be written locally
$$D= \sum _{|j|\leq m}  a_j(x) \d_x^j$$ 
where $j$ is a multi-index of length at most $m$. Its symbol is the polynomial obtained by freezing the coefficients $a_j$
 and taking the Fourier transform 
$$ d (x,p) = \sum _{|j|\leq m}  a_j(x) (ip)^j\,.$$
This definition can be extended to pseudo-differential operators (involving typically inverse derivatives)
 by considering symbols which are more general functions
than polynomials. The principal symbol $h$   of a pseudo-differential operator
$H$  of order $m$  is  defined  in terms
of the action of $H$ on fast oscillating functions $ae^{i\tau S}$, where $a$ is smooth and $S'\ne 0 $ on the support of $a$, by 
\[  H (ae^{i\tau S})(x)=\tau^m
h(x, S'(x) ) a(x)e^{i\tau S (x)}+ O\left( \tau ^{m-1} \right) \]
The function $h$  is a  smooth homogeneous function of degree $m$  on 
$T^\star X \setminus 0$, which is well defined independently of the choice of local coordinates.

The key property is   that the
 principal symbol $h\equiv h(x,p) $ of a pseudo-differential operator $H $, even though it does not determine $H$ completely, contains almost 
 all the information, up to smoother error terms. More precisely, one defines classes of pseudo-differential operators
 according to the order $m$ of their symbol
$$ S_m =\{ \sigma = \sigma (x,p) \in C^\infty( T^*X \setminus X\times \{0\} )  \,/\, |\d^\alpha _x \d^\beta_p \sigma|
\leq C_{\alpha,\beta}(1+ |p|)^{m-|\beta|}\}\,.$$
Then, if $H$ is a pseudodifferential operator of principal symbol $h$ and order $m$, we can define another operator $Op(h)$
using the Weyl quantization 
$$
 {\rm Op}(h) u(x ) := \frac1{(2\pi )^{{\rm dim}(X)}} \int e^{i(x-x') \cdot p} h \left(\frac{x+x'}2,p\right) u(x') \: dx' dp
$$
and compare both operators: we find that
$H-{\rm Op} (h)$ is of order at most $m-1$.

Furthermore, one can prove that the symbol of the commutator $i[H_1,H_2]$ is the Lie bracket $\{h_1,h_2\}$, so  that 
$${\rm order} ([H_1,H_2]) \leq {\rm order} (H_1) + {\rm order}(H_2) - 1.$$
By iterating this kind of estimates, we then see that the symbolic calculus allows to obtain expansions up to any regularizing order.
For a nice introduction to this subject, we refer to \cite{alinhac-gerard, Du, XSR}.

\medskip
\noindent
$\bullet$
{\bf Typical operators with continuous spectrum} are multiplications by smooth functions $\nu = \nu(x)$
 (the spectrum is then the range of $\nu$). In this case, the classical dynamics is defined by the Hamiltonian $h(x,p) = \nu(x)$
$${dx \over dt} = {\d h \over \d p} = 0,\quad {dp \over dt} = -{\d h \over \d x} =-\nu'(x) = - \nu'(x_0)\,,$$
so that $\pm p$ is a Lyapunov function, going to infinity as $t\to \infty$.

In contrast,
 operators with discrete spectrum (called ``elliptic operators'') have eigenfunctions of finite energy,
 meaning that the classical dynamics has compact invariant sets.
 Therefore, in this case, we cannot have any escape function as no quantity will go to infinity.

This is of course not a proof but at least a strong indication that  the fact that  a pseudo-differential operator $H$
 has continuous spectrum is related to the existence of an escape  function $d$ 
  for the Hamiltonian dynamics associated to its principal symbol $h$
$${d \over dt} d(x(t), p(t)) \geq \gamma >0$$
denoting by $(x(t), p(t))$ the trajectories of the classical dynamics:
$${dx \over dt} = {\d h \over \d p} ,\quad {dp \over dt} = -{\d h \over \d x}\,.$$

\medskip
\noindent
$\bullet$
A nice mathematical tool to study the continuous spectrum of a skew-symmetric operator $iH$ (such as our wave operator)  is the {\bf conjugate operator method } that we will present in Section 5.
The crucial point in this theory is the existence of  a self-adjoint operator $D$ satisfying the commutator estimate
\begin{equation}
\label{commut-cond}
 [H,D] \geq \gamma I +K \hbox{ with } \gamma>0,  \hbox{ and $K$ a compact operator}\,,
 \end{equation} 
Note that the existence of a conjugate operator $D$ for the wave operator $H$ can be translated in terms of the classical dynamics associated to $H$: it is roughly equivalent to the existence of an escape function $\psi$  since the condition (\ref{commut-cond})  can be expressed in terms of the Lie bracket
\begin{equation}
\label{bracket-cond}
\{h, d\} \geq \gamma >0\,.
\end{equation}

Under this condition, the conjugate operator method shows how to  extend  the resolvent $(H-\lambda\pm i\eps)^{-1}$ which is defined for $\eps>0$ and $\eps <0$, up to $\eps =0$ 
$$ (H-\lambda\pm i0)^{-1} =\lim_{\eps \to 0^+} (H-\lambda\pm i\eps)^{-1}\,.$$
 As a consequence, on can prove that  the spectral decomposition of a smooth function $f$, i.e. the measure $\nu$ defined by
$$ f =\int d\nu(\lambda) , \qquad \chi(H) f = \int \chi(\lambda) d\nu(\lambda) $$
is a regular function of $\lambda$ given by
$$ \nu(\lambda) = {1\over 2i\pi} \left( (H-\lambda- i0)^{-1} f - (H-\lambda+ i0)^{-1} f\right) \,.$$
The proofs of this result are not very complicated, but it is not completely obvious to get a good intuition of this property.

\subsection{Classical dynamics for Hamiltonians of degree 0}

The last important ingredient we will use is the fact that the existence of an escape function
 is a generic property of Hamiltonians of degree $0$ in dimension $2$. 

\begin{defi}
 A smooth  homogeneous function of degree $1$, 
 $d:\Sigma_{\omega _0}\rightarrow \R $, is 
called an {\rm escape function} if  the Poisson bracket
$\{h,  d \} $, which is homogeneous of degree $0$,  is strictly positive
on the energy surface $\Sigma_{\omega _0}= h^{-1}(\omega_0)$.
\end{defi}

\medskip
\noindent
$\bullet$ Let us first have a look at the classical dynamics of internal waves in the trapezoidal geometry. The dynamical equations
$$\left\{\begin{aligned}
& \frac{dx_1}{ dt}  = \frac{p_3^2}{ |p|^3}, \quad &\frac{dp_1}{ dt} = 0,\\
 & \frac{dx_3}{ dt}  =- \frac{p_3p_1}{ |p|^3}, \quad  &\frac{dp_3}{ dt } = 0\,.
  \end{aligned}\right.
  $$
have to be supplemented with boundary conditions. To determine the reflection laws, we look at the solutions of the Boussinesq system in some half space delimited by a slope tilted of an angle $\alpha$ with respect to the horizontal
    $$ x_1 \sin\alpha + x_3 \cos\alpha = C\,.$$ 
    We seek these solutions in the form of an incident wave  propagating with an angle $\phi$ with respect to the horizontal (recall that the direction of propagation is orthogonal to the wavenumber) plus a reflected wave
$$\begin{aligned}
W=\lambda  U(p) + \mu V(p) \\
W'=\lambda ' U(p') + \mu 'V(p') 
\end{aligned}
$$
with
$$
\begin{aligned}
U(p)= \left( \begin{array}{c} 1 \\0\\0 \end{array} \right) \cos (\omega t + p\cdot x)
- \left( \begin{array}{c} 0 \\ -Np_3/(\rho' |p|) \\  Np_1/(\rho' |p|)\end{array} \right) \sin  (\omega t + p\cdot x)  \\
V(p)= \left( \begin{array}{c} 1 \\0 \\0\end{array} \right) \sin (\omega t + p\cdot x)
+ \left( \begin{array}{c} 0 \\ -Np_3/(\rho' |p|) \\  Np_1/(\rho' |p|)\end{array} \right)  \cos  (\omega t + p\cdot x) 
\end{aligned}
$$
   We then obtain necessary conditions on the wavenumber of the reflected wave
  $$ 
  \left\{\begin{aligned}
  &\frac{(p'_1)^2 }{ |(p')^2|} =  \frac{p_1^2 }{ |p|^2}  = \sin^2 \phi \hbox{ coming from the conservation of energy}\\
  & (p'_1- p_1) \cos \alpha - (p'_3 - p _3 ) \sin \alpha = 0 \hbox{ for the  phase  to be constant on  the slope}
   \end{aligned}\right.
   $$
   as well as some polarization conditions to determine $(\lambda',\mu')$ from $(\lambda,\mu)$.

   \bigskip
   We will  use the following reduction
 to represent the trajectories. Specular reflection on the horizontal and vertical boundaries is equivalent to 
free propagation in a domain which is extended by symmetry. Combining two successive reflections with respect to horizontal boundaries, we get a simple vertical  translation, which means that  we obtain a periodic structure with respect to the vertical variable. We then identify the points of the horizontal boundary in the extended domain.

\begin{figure} [h] 
\centering
\includegraphics[width=7cm]{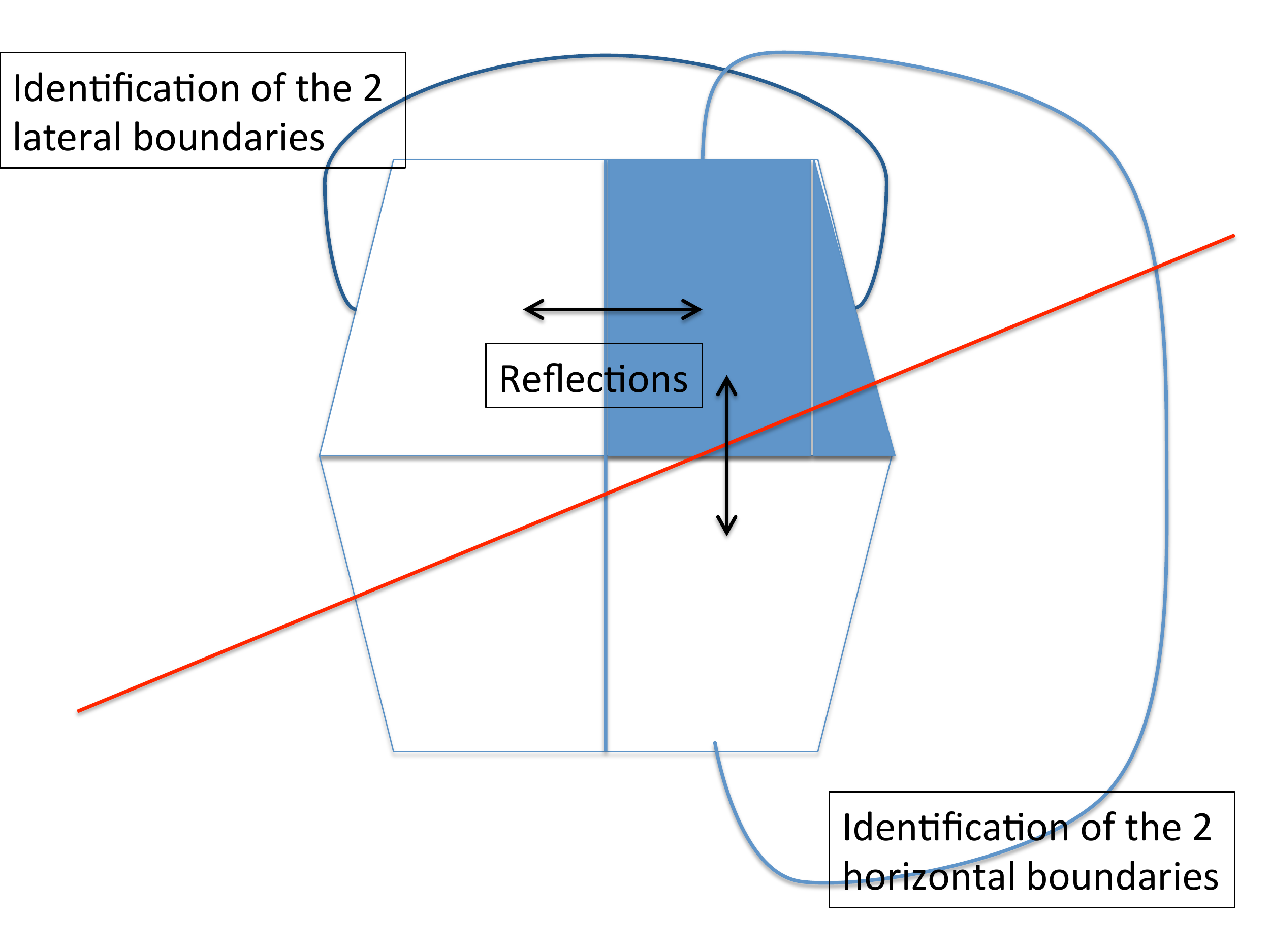} 
\caption{\small  
The image method leads to a singular dynamics  in a periodic  domain.}
\label{fig: rays}
\end{figure}

When a  trajectory exits the domain on the left (resp. on the right), it re-enters on the right (resp. on the left) at the symmetric point.
One can then identify the points of the lateral boundary, and get a  (non smooth) dynamics on a torus, along a fixed vector field.
We have actually four copies of this dynamics, corresponding to the four possible directions given by $p_1^2/|p|^2 = \sin^2 \phi$.

\medskip
Assume that the angle  $ \phi$  lies in the interior of a band of the bifurcation diagram, i.e.  that 
 there exists an open set $I$ containing $\phi$  on which  $\rho_{\alpha, \phi}\in \Q $ is constant
and the periodic points are hyperbolic.

 All trajectories of the semiclassical dynamics with energy $\omega_0 = |\sin \phi|$ then converge to some attractor, corresponding to some fixed point of the (iterated) Poincar\'e return map $\pi_{\alpha,\phi_0}$. 
The wavenumber is exponentially increasing since it is multiplied by  the same constant (depending only on $\alpha$ and $\phi$) at each iteration
$$p_3^{(n)} \sim C \left( \frac{\sin (\phi+\alpha)}{ \sin(\alpha-\phi) } \right)^n \,,$$
while the distance to the limit cycle is exponentially decreasing with $n$
$$| x_3^{(n)} - \bar x_3 | \leq  C \left( \frac{\sin (\alpha-\phi )}{ \sin(\phi+\alpha) } \right)^n \,,$$
 Since the group velocity is proportional to  $|p_3^{(n)}| ^{-1}$, the return time, i.e. the time $t_n$ needed to run through the approximate cycle (from $x_3^{(n)}$ to $x_3^{(n+1)}$)  is exponentially increasing with $n$
$$c \left( \frac{\sin (\phi+\alpha)}{ \sin(\alpha-\phi) } \right)^n \leq   |t_n| \leq  C\left( \frac{\sin (\phi+\alpha)}{ \sin(\alpha-\phi) } \right)^n\,.$$
In particular, the wavenumber is of the order of $O(t)$ and is therefore an escape function.

\begin{Rmk}  In this example, we have exhibited an escape function, in the sense that it converges to infinity as $t\to \infty$ (with an average growth which is linear), but this function is 
piecewise constant with jumps. This weak notion  of  escape function and the lack of regularity will be major obstacles to apply the general theory.
\end{Rmk} 

\medskip
\noindent
$\bullet$ In the general case, under a suitable regularity assumption on the Hamiltonian $h$, we will actually prove that there exists a change of variables (also called  normal form) such that $h-\omega_0$ can be locally represented by the following simple toy model on $\T_x \times \R_y $~:
$$ h_0(x,y, \xi,\eta)  = {\xi \over \eta} - \lambda y \,.$$
The classical dynamics admits two first integrals $h_0$ and $\xi=\xi_0$. It is given by
$$
{dx\over dt } = {1\over \eta}, \quad {dy \over dt} = -{\xi \over \eta^2}\qquad {d\eta \over dt } =  \lambda\,.$$
We then have, for $t\geq 0$,
$$ x = x_0 + \frac1\lambda \log \left(1 +{\lambda t \over \eta_0}\right) \hbox{ mod}(\Z), \qquad  y = {\xi_0\over \lambda (\eta_0+\lambda t)} , \qquad \eta = \eta_0 +\lambda t\,.$$

Let us denote by $\gamma$ the curve $y=0$ in $\T_x \times \R_y$. We see that the trajectory spirals around $\gamma$ infinitely many times with a speed going to 0.
Let us then look at the Poincar\'e map $\pi$ associated to the section $S= \{(x,y,\xi, \eta) \,/\, x=0, \xi= \lambda y \eta, \eta>0\}$ with the symplectic form given by $dy \wedge d\eta$:
$$\pi (y,\eta) = ( e^{-\lambda}  y ,e^{\lambda}  y)$$
and we have a geometric progression of momenta $\eta$, which is an escape function.

This model is therefore a smooth version of the dynamics associated to internal (or inertial) waves. We will show in Section 6 that, under rather general assumptions (stated in the next paragraph), one can always reduce  the study of (smooth) Hamiltonians of degree 0 to this simple model.


\section{Main results}

As explained in the previous paragraph, the mathematical setting we consider in this paper  reproduces the important features
 of the inertial and internal wave operators in domains with topography, but with  more regularity
in order that  techniques of pseudo-differential calculus can be used.

\bigskip
More precisely, we consider a general scalar equation of the form 
\begin{equation}
\label{forced}
\frac{1}{i} \d_t u + Hu = f e^{-i\omega_0 t}
 \end{equation}
 on a 2D torus  $X$  equipped with a smooth density $dx$,
where $fe^{-i\omega_0 t}$ is a periodic  smooth forcing,
 and $H$ is a bounded self-adjoint  operator on $L^2(X,dx)$ satisfying the following assumptions.
$$\vbox{ \hbox{\bf (M0)  $H$ is a pseudo-differential operator 
 with a smooth }\hbox{\bf   principal symbol $h : T^*X\setminus\{0\} \to \R$ positively homogeneous of degree  $0$}
 \hbox{\bf  
 and a vanishing  subprincipal symbol.} }$$

Let us recall for the sake of completeness that the  sub-principal symbol $h_{-1}$ of 
$H$ is the  smooth homogeneous function of degree  $-1$ on 
$T^\star X \setminus 0$ defined  by:
\[ \langle H (ae^{i\tau S})|ae^{i\tau S}\rangle_{L^2 (X,dx)}
= \int _X |a(x)|^2 \left(h(x, S'(x) )+\tau^{-1} h_{-1}(x,S'(x) )\right) dx + O\left( \tau ^{-2} \right) \]
Note that the function $h_{-1}$    depends on the measure $dx$.

The assumption that $h$ is homogeneous of degree 0  implies that any energy shell $\Sigma_{\omega}=h^{-1} (\omega)$ is conic:
$$ \forall (x,p) \in \Sigma_{\omega},  \forall \lambda \in \R^+, \quad (x,\lambda p) \in \Sigma_{\omega}\,.$$
We  will then denote by $Z$ the oriented manifold of dimension $2$ which is  the quotient of the conic energy shell 
$\Sigma _{\omega _0}$ by the positive homotheties
$$Z=\{ (x, e) \in X \times S^1 \,/\, h(x,e) = \omega_0\}\,.$$
We can think of $Z$ as the boundary at infinity of the energy shell.

\bigskip
For the sake of simplicity, we now  introduce  some   geometrical  assumptions on $\Sigma_{\omega_0}$ (and $Z$) in order to avoid singularities. These assumptions can actually be relaxed as shown  by the first author in the recent paper \cite{CdV}.
 $$\vbox{ \hbox{\bf (M1) The energy shell $\Sigma _{\omega_0}$ is nondegenerate and }\hbox{\bf 
 the canonical projection $\pi :Z \rightarrow X$ is a finite covering of degree $n$.}}$$
 
The non degeneracy assumption means  that $dh\neq 0$ on $\Sigma_{\omega_0}$, it  is a generic assumption on  the frequency $\omega_0 $
  thanks to Sard's theorem.
The second assumption means that,
for each $x\in X $, the pre-image of $x$ by the canonical projection $\pi$ has exactly $n$ elements, 
  and the directions in $\pi ^{-1}(x)$ are smoothly dependent on $x$.
Note that, in the case of internal waves in a trapezium, the number of admissible directions $n$ at each regular point $x$ is 4,
 but there are singularities.

By definition, the Hamiltonian vector field $X_h$ is tangent to $\Sigma _{\omega _0 }$ and homogeneous of degree $-1$. This implies that the
oriented direction of $X_h$ induces a field  of oriented directions on $Z$.
We therefore say that  $Z$ is equipped with a 1D-foliation, denoted by ${\cal F}$.
The leaves of this foliation correspond to  the orbits of the Hamiltonian dynamics in $Z$ (which are defined modulo  a positive change of time).
The compact leaves correspond to  periodic orbits.

Note that the foliation is on $Z$. The projection on $X$ is a multiple foliation, meaning
that at each point of $X$ there are $n$ distinct oriented directions.

\begin{Lem}
 Under the assumption (M1), the foliation  ${\cal F}$ is non singular. 
 \end{Lem}
 
 \begin{proof}
Assume that the foliation  ${\cal F}$ is  singular.

Since $\Sigma _{\omega_0}$ is nondegenerate, the foliation can be singular only at the points where $X_h $ is parallel
to the cone direction. The projection of the foliation ${\cal F}$ on $X$ is generated by the vectors
$(\partial_{p_1} h,\partial _{p_2} h )$. Therefore the foliation ${\cal F}$ is singular if and only if  this vector vanishes. 

But then
the tangent space to $\Sigma _{\omega _0 }$, which  is defined by the non trivial equation 
$(\partial _{x_1} h) \, dx_1 + (\partial _{x_2} h )\, dx_2 =0 $, does not project in a surjective way on the tangent space of $X$. This contradicts our
assumption that the canonical projection is a finite covering of degree $n$.

Note that we also retrieve the property  that the wave number $p $ is orthogonal (for the duality) to the direction of the projection of
$X_h$ on $X$. This is due to the Euler relation
$p_1 \partial _{p_1}  h +p_2 \partial _{p_2}  h=0 $. Another way to interpret this last condition is to say  that the cones generated by the leaves of $ {\cal F}$ are Lagrangian. 
\end{proof}

\bigskip
 The last assumption is on the dynamics of ${\cal F}$:
$$\hbox{\bf (M2) The foliation ${\cal F}$ is Morse-Smale},$$
which is a generic condition.
Let us recall what  it means: 

\begin{itemize}
\item there is a finite number of compact leaves (diffeomorphic to circles), also called cycles in the sequel. 
\item  Each compact leaf is hyperbolic
(the corresponding linear Poincar\'e map - which is defined intrinsically -  is  strictly expanding or contracting). 
\item And all other leaves are accumulating only 
along two of the previous 
closed leaves at $\pm \infty $. 
\end{itemize}


Note that conditions $(M1)$ and $(M2)$ are stable under small perturbations, therefore are still verified for
 energy levels close to $\omega_0$.

\begin{Rmk}
The regularity assumption is encoded in $(M0)$. In particular, at this stage,
 even though we can capture the effect of the zero flux condition for internal waves in a model without boundary
 (see Fig. \ref{fig: rays}), this model is not smooth enough to enter in this class of operators.

- If the boundary is a polygon, we have seen that the Poincar\'e map is only piecewise affine.
 This difficulty could be removed by considering a smooth domain, which leads to a foliation with singular points
 (corresponding to critical angles). But our results can actually be extended to singular foliations (see \cite{CdV}).

- A more serious difficulty comes from the fact that the wavenumber $p$ jumps at each return time,
 which means that there is no smooth normal form which conjugates the geometric object given by the foliation
 and the Hamiltonian dynamics.

To tackle the original problem, we would probably need to introduce another covering to take into account this additional complexity.
\end{Rmk}

\bigskip
In this abstract setting, we can now formulate our result describing the long time behavior of the forced system.
\begin{Thm}\label{main-thm}
Let $H$ be a pseudo-differential operator such that assumptions $(M0)(M1)(M2)$ are satisfied for any $\omega_0$ in some  (open) interval $I$.

Then $H$ has at most a finite number of eigenvalues in $I$, and for any $\omega_0\in I\setminus \sigma_{pp}(H)$ and any $f\in C^\infty(X)$, the solution to the forced  equation
$$-i \d_t u + Hu = f e^{-i\omega_0 t}, \qquad u_{|t=0} = 0
$$
  can be decomposed in a unique way as
\[ u(t)e^{i\omega_0 t} = u_\infty+ b(t)+\epsilon (t) \]
where 
\begin{itemize}
\item  $u_\infty= (H-\omega _0 -i0)^{-1} f $ belongs to the Sobolev spaces $ \cH^{-1/2-0}$ and  is not in $L^2$ except if it vanishes;
\item 
 $b(t)$ is a bounded  function  with values in $L^2$ whose time Fourier transform vanishes near $0$;
\item   $\epsilon (t) $ tends  to $0 $ in $\cH^{-1/2-0}$.
\end{itemize}
The   wavefront  set of the limiting distribution $WF(u_\infty)$ is 
contained in the  cones generated by the  stable cycles of ${\cal F}$.

Furthermore, the energy $\| u(t) \| ^2 _{L^2 (X,dx)}$ grows linearly except if $u_\infty $ vanishes. 

\end{Thm}

 Roughly speaking, the wave
 front set $WF(u_\infty)$ characterizes the singularities of the generalized function $u_\infty$, not only in space, but
 also with respect to its Fourier transform at each point. In more familiar terms, $WF(u_\infty)$ tells not only where the
 function $u_\infty$ is singular, but also how or why it is singular, by
 being more exact about the direction in which the singularity occurs. 

More precisely, if $u$ is a Schwartz distribution on $\R^d$,   $(\bar x, \bar p )\notin WF(u)$ means that  there exists a test function $a$ with $a(\bar x)\ne 0$ so that the Fourier transform of
$au$ is fast decaying in $p$ in some conic neighborhood of $\bar p $, or in other words that there exists a pseudo-differential  operator $A$ elliptic 
at the point $(\bar x, \bar p)$ so that $Au$ is smooth. This definition is completely intrinsic, and can be extended for a  smooth manifold $X$. We refer to \cite{XSR, bony} for a simple introduction to the wavefront set.

We will give in Theorem \ref{theo:OIF} a much more precise description of $u_\infty $ as a Lagrangian state (or Fourier integral distribution)
 associated to the
previous conic Lagrangian manifolds.

Note that Dyatlov and Zworski have obtained recently \cite{DZ} an alternative proof of a slightly weaker result, based only on microlocal techniques and radial estimates.

\section{Quasi-resonances and long time behaviour}\label{longtime-sec}

We first show how the  (local)  spectral representation of $H$ can be used to describe the long time behaviour
 of the solution  to the forced equation. 
 The solution $u$ of (\ref{forced}) is given  by
\[ u(t) = \frac{e^{-it\omega_0 }-e^{-it{H}}}{H-\omega_0 }f \]
Assume that $H$ has some continuous spectrum around $\omega_0$,  and that we know the spectral decomposition of $f$ with respect to  $H-\omega_0$:
$$ \chi (H-\omega_0) f = \int _\R \chi( s) d\nu(s)\,.$$
We thus have that
$$ u(t) =e^{-it\omega_0} \int  \frac{1-e^{-its}}{s } d\nu (s) \,.$$

\bigskip
We will then use  a  general result on  oscillating integrals, extending the computations of Section 2.1:

\begin{lem}\label{lem-int}
Let $\cB, \cB_0 \subset \cB$ be two Banach spaces. 
Let  $\nu =\nu _{\rm ac} + \nu_{\rm sing } $ be a compactly supported Radon measure on $\R$ with values in $\cB$ such that
\begin{itemize}
\item  the absolutely continuous part   $\nu _{\rm ac} $ has a density  $m$  which is H\"older continuous  $C ^{0,\mu} $ ($\mu >0$) near $0$;
\item  the singular part $\nu_{\rm sing }$ has values in $\cB_0$ and is supported outside from 0.
\end{itemize}
Then the oscillating  integral 
$ I(t)= \langle  (1- e^{-its}) s^{-1}| d\nu (s) \rangle $ can be decomposed in a unique way as
\[ I(t)= I_\infty   + b(t)+\epsilon (t)  \]
where
$ I_\infty =\langle  (s-i0)^{-1}| d\nu (s) \rangle \in \cB$, 
   $\epsilon (t) \rightarrow 0 $ in $\cB$ as $t\rightarrow \infty $, 
and $b $ is bounded with values in $\cB_0$  and inverse Fourier transform
supported on   $\supp( \nu_{\rm sing }) $.
\end{lem}

\begin{proof}
The proof  is actually a simple calculation  on distributions.
Let $a>0$ be such that 
$$ [-a, a] \cap \supp(\nu_{sing}) = \emptyset\,.$$
Since  $\nu$ is compactly supported, there is no issue of convergence at infinity.
We then split the integral into two parts:
$$\begin{aligned}
I_1(t) =\int _{|s|< a}  \frac{1- e^{-its}}{s} d\nu_{ac} + \int _{|s|\geq a}Ê \frac{1}{s} d\nu_{ac}   \\
 I_2(t) = -  \int _{|s|\geq a}Ê \frac{e^{-its}}{s} d\nu_{ac}+   \int _{|s|\geq a}Ê \frac{1- e^{-its}}{s} d\nu_{sing} \,.
 \end{aligned}
 $$ 
 We check that the second term in $I_2(t)$ satisfies the properties required for $b(t)$, while the first term in $I_2(t)$ tends to $0$  thanks to the Riemann-Lebesgue theorem.

We have then to study $I_1(t)$. As $m$ is of class $C^{0,\mu} $ near 0, we have that  $m_1(s)=( m(s)- m(0)) /s $ belongs to $L^1$, so that 
\[ 
\begin{aligned}
I_1 (t)=&\int_{|s|< a} \frac{1-\cos (ts)}{s} (m(0)+sm_1(s) ) ds \\
&+ i\int_{|s|< a} \frac{\sin (ts)}{s}  (m(0)+sm_1(s) ) ds+ \int _{|s|\geq a}Ê \frac{1}{s} d\nu_{ac}
\end{aligned} \]
The second integral converges to $i\pi m(0)$. 
Using the parity of $\cos $ and  the Riemann-Lebesgue theorem, we obtain that the first integral tends to  
$$\int_{|s|< a}  \frac{m(s)-m(0)}{s}  ds={\rm p.v.}\left( \int _{|s|< a} \frac{m(s)}{s}  ds \right)\,.$$
Adding the third term, we can remove the truncation.

In the limit, we get finally
$${\rm p.v.}\left( \int  \frac{m(s)}{s}  ds\right)+i\pi m(0)=\langle (s-i0)^{-1}|m \rangle $$
 by the Sokhotski-Plemelj theorem. 
\end{proof}

\section{The conjugate operator method}\label{mourre-sec}

In order to describe the response of the system to some monochromatic forcing at frequency $\omega_0$, we  therefore need  to get  the spectral structure of $H$ close to this frequency, assuming that it is non degenerate in the sense of assumptions (M1)(M2). We will first show how to define the spectral density, assuming the existence of a conjugate operator.

\subsection{A short review on Mourre theory}

Let $H$ be a self-adjoint operator on some Hilbert space, say $L^2$. 
Here we will further assume that $H$ is bounded, as it is the case in the application we have in mind.

\begin{Def}
Let $D$ be a self-adjoint (unbounded) operator. For any operator $A$, we denote by $i[A,D]$ the closure of the form $i(AD-DA)$ defined on $\cD(A)\cap \cD(D)$.
\begin{itemize}
\item We  say that $H$ is $n-$smooth with respect  to $D$ if the iterated brackets 
 ${B}_1:= i[ {H},D] $ and  ${B}_k:=[ {B}_{k-1},D]$ are bounded up to $k=n$.
\item We define also the Sobolev scale $\cH^s$ ($s\in \R$) associated to $D$ by 
$$\cH^s= \{ u \in L^2 \,/\, (1+D^2)^{s/2} u \in L^2\} $$ for $s\geq 0 $, 
and as the dual of $\cH^{-s}$ for $s<0$. 
\end{itemize}
\end{Def}

The main result of the conjugate operator theory (see \cite{mourre, JMP}) is the following theorem:

\begin{Thm}[Mourre]  \label{mourre-thm}
Let us assume that  $H$ is $n-$smooth with respect to $D$ with $n\geq 2$, and that we have the following commutator estimate:
 for $\chi, \bar \chi   \in C_0^\infty (\R,\R^+ )$ with $\chi \bar \chi  =\bar \chi $,  
\begin{equation}  
\label{commutator-hyp}
\chi(H)  {B}_1 \chi(H) \geq  \alpha \bar \chi  (H)  +K 
\quad \hbox{ for some compact operator $K$}.
\end{equation}
Then, for any closed interval $I \subset \supp (\bar \chi)$
\begin{itemize}
\item[(i)] $H$  has a finite set $\sigma _p(H)$  
of eigenvalues of finite multiplicity  in $I$;
\item[(ii)]   the resolvent
$({H}-z)^{-1} $ defined for $\Im (z) \ne 0$ admits boundary values at the points  $\omega \in I \setminus \sigma_p (H) $
in the space $O_s:= L(\cH^s , \cH^{-s})$ for $s>1/2 $.
\item[(iii)]  the boundary values $({H}-\omega \pm i0 )^{-1} $
are H\"older continuous $C^{0,\mu} ( I \setminus \sigma_p (H), O_s)$ with $s>1/2$ and $\mu = (2s-1)/2s$. 
\item[(iv)] the boundary values $({H}-\omega \pm i0 )^{-1} $
admits continuous derivatives of order $n$ in the spaces $O_s$ with $s>n- 1/2$. 
\end{itemize} 
\end{Thm}

\begin{proof} For the sake of completeness, we give here a sketch of proof of a slightly simpler result, which turns out to be quite simple in our case since $H$ and $B_1$ are bounded operators in $L^2$.

\medskip
\noindent
$\bullet$ The first step is to prove that the discrete spectrum in the interval $I$ is finite. If it is not, there exists a sequence of orthonormal eigenfunctions $\phi_n$ with $H\phi_n = \omega_n \phi_n$.
By the commutator estimate (\ref{commutator-hyp}), we then get
$$ 0= \la \phi_n, B_1 \phi_n\ra \geq \alpha \| \phi_n\|^2 + \la \phi_n, K\phi_n\ra\,.$$
Since the $\phi_n$ are orthogonal, $\phi_n\rightharpoonup 0 $  weakly in $L^2$, and since $K$ is compact, 
$$\lim _{n\to \infty}  \la \phi_n, K\phi_n\ra = 0\,.$$
We  obtain a contradiction.

\medskip
\noindent
$\bullet$ The second step is to define some suitable approximation for the resolvent, when
 $\omega_0$  is not an eigenvalue of $H$. 
Define $P_\delta$ to be the spectral projector of $H$ on $[\omega_0-\delta, \omega_0+\delta]$. As $\delta \to 0$, $P_\delta$ converges weakly to 0, and since $K$ is compact,
$ K P_\delta \to 0$. One can then find $\delta$ small enough so that $\pm 2 P_\delta KP_\delta  \leq \alpha  P_\delta^2 $, from which we deduce that 
\begin{equation}\label{alpha}
P_\delta B_1 P_\delta \geq \frac{\alpha}{2} P_\delta^2 \,.
\end{equation}
In the sequel we will remove the subscript $\delta$ and call $P$ this spectral projection.

We then define
$BB^* = P_\delta B_1 P_\delta $ and 
$$ G_z(\eps )= (H-z- i\eps BB^*)^{-1}, \quad F_z(\eps) = A_\eps G_z(\eps)  A_\eps \,,$$
where $A_\eps = |D+i|^{-s} |\eps D +i|^{s-1} $ for some $s>1/2$ and $\eps >0$.
When $\eps$ and $\Im z$ have the same sign, as $H$ is self-adjoint,
$$\begin{aligned}
 \| (H-z- i\eps BB^*) \varphi \|^2 &= \| (H- \Re z ) \varphi\|^2 + \| (\Im z + \eps BB^*) \varphi \|^2 - 2 \Im ((H- \Re z ) \varphi, \eps BB^* \varphi) \\
 & \geq (\Im z)^2 \| \varphi\|^2 
 \end{aligned}
 $$
so that $H-z- i\eps BB^*$ is injective with closed range in $L^2$. Since its adjoint is also injective, the range is actually $L^2$.
 By the open mapping theorem, its inverse  $G_z(\eps)$  exists as a bounded operator. 
 
Using (\ref{alpha}) and the fact that $\Im z \geq 0$, we  also have that 
$$
\begin{aligned}
\| PG_z(\eps) A_\eps  \|^2& = \| A_\eps  \bar G_z(\eps )P^2 G_z(\eps) A_\eps  \|\\
&\leq \frac{2}{\alpha \eps}  \| A_\eps \bar G_z(\eps )(\eps BB^* +\Im z) G_z(\eps) A_\eps  \|\\
& \leq  \frac{1}{\alpha \eps} \| A_\eps  ( \bar G_z(\eps) - G_z(\eps)) A_\eps  \| \\
& \leq \frac{2}{\alpha \eps} \| A_\eps  G_z(\eps) A_\eps  \| \leq  \frac{2}{\alpha \eps} \| F_z(\eps)\| \\
\end{aligned}
$$
Finally, using the spectral localization, we have that for any $z$ such that $\Re z \in [\omega_0 - \delta/2, \omega_0+\delta/2]$
\begin{equation}
\label{bound1}
\| (I- P) G_z(\eps) A_\eps \| \leq C_\delta.
\end{equation}
In particular, one has 
\begin{equation}
\label{bound2}
\| G_z(\eps) A_\eps  \| \leq   C_{\alpha, \delta } \left(1+ \left(  \frac{ \|F_z(\eps)\|}{ \eps} \right) ^{ \frac{1}{ 2}}\right) \,.
\end{equation}
Note that, at this stage, we did not use the specific form of $A_\eps$. In particular, we have that
\begin{equation}
\label{bound3}\| G_z(\eps) \| \leq \frac{C}\eps \,.
\end{equation}

\medskip
\noindent
$\bullet$ The third step is to obtain a differential inequality in order to control the dependence with respect to $\eps$.
Let us compute the derivative of $F_z(\eps)$ with respect to $\eps$:
$$
\begin{aligned}
\frac{d F_z(\eps) }{ d\eps}& = {dA_\eps \over d\eps } G_z(\eps)  A_\eps +  A_\eps G_z(\eps) {dA_\eps \over d\eps }\\
&\quad +A_\eps  G_z(\eps) P[D,H]P  G_z(\eps)  A_\eps  \,.
 \end{aligned}$$
A straightforward computation shows that
$$P[D,H]P = [D,H-z-i \eps BB^*] +[D, i\eps BB^*]  + (P-I) [D,H] P + [D,H] (P-I)$$
from which we deduce that 
$$
\begin{aligned}
\frac{d F_z(\eps) }{ d\eps} &={dA_\eps \over d\eps }  G_z(\eps)  A_\eps + A_\eps G_z(\eps) {dA_\eps \over d\eps} + A_\eps \left( -DG_z(\eps)+ G_z(\eps ) D\right) A_\eps\\
& \quad   +\eps A_\eps  G_z(\eps)[D, i B^*B]   G_z(\eps) A_\eps \\
& \quad +  A_\eps  G_z(\eps)(P-I) [D,H] P  G_z(\eps) A_\eps +  A_\eps  G_z(\eps) [D,H] (P-I ) G_z(\eps) A_\eps 
\end{aligned}
$$
Since the commutators $[D, BB^*]$ and $[D,H]$ are bounded, and that
 \begin{equation}
\label{bound4}
\begin{aligned}
 \left\| {dA_\eps \over d\eps} \right\| = (1-s)  \left\| |D+i|^{-s}  |D|  |\eps D+ i|^{s-2 }\right\| \leq (1-s) \eps ^{s-1} \,,\\
 \| A_\eps D \| =\left\| |D+i|^{-s}  |D|  |\eps D+ i|^{s-1 }\right\| \leq \eps ^{s-1} \,.
 \end{aligned}
 \end{equation}
  we deduce from the
 a priori estimates (\ref{bound1})(\ref{bound2}) that 
\begin{equation}
\label{diff-ineq}
\left\| \frac{d F_z(\eps) }{ d\eps}\right\| \leq  C\eps ^{s-1} \left(1+ \left(  \frac{ \|F_z(\eps)\|}{ \eps}\right) ^{ \frac{1}{ 2}} + \| F_z (\eps) \|  \right)\,.
\end{equation}

\medskip
\noindent
$\bullet$
The fourth step is then  to extend the resolvent as $\eps \to 0$ by a continuity argument.
By (\ref{diff-ineq}) and (\ref{bound2}), we get that
$$\left\| \frac{d F_z(\eps) }{ d\eps}\right\| \leq  C\eps ^{ s-3/2}\left( \eps + \|F_z(\eps)\|\right) ^{1/2} \,,$$
which upon integration shows that $ \| F_z(\eps)\| $ is uniformly bounded with respect to $\eps$, and that 
$ \| F_z(\eps) - F_z(\eps') \| $ goes to 0 as $\eps, \eps' \to 0$.

This proves in particular  that $({H}-z)^{-1} $ defined for $\Im (z) > 0$ admits boundary values at the points 
 $\omega \in [\omega_0 - \delta/2, \omega_0+\delta / 2] $
in the space $O_s:= L(\cH^s , \cH^{-s})$. Furthermore,
\begin{equation}
\label{bound5}
\left\| F_z(\eps) -F_z(0) \right\| \leq C\eps^{2s-1}\,.
\end{equation}

\medskip
\noindent
$\bullet$ To obtain the H\"older continuity with respect to $z$, we compute
$$\frac{d F_z(\eps) }{ dz} = A_\eps  G_z(\eps)   G_z(\eps)  A_\eps \,.$$
Using the uniform bound on $F_\eps(z)$ and (\ref{bound2}), we obtain that
$$ \left\| \frac{d F_z(\eps) }{ dz}\right\| \leq \frac C\eps \,.$$
Then,
$$ \| F_z(\eps) - F_{z'} (\eps ) \| \leq \frac C\eps |z-z'| \,.$$
Combining this estimate with  (\ref{bound5}) and choosing $\eps ^{ 2s} =|z-z'|$, we get
$$ \| F_z(0) - F_{z'} (0 ) \| \leq  C |z-z'|^\mu \quad \hbox{ for }Ê\mu = (2s-1)/(2s)\,.$$

\medskip
The additional regularity is obtained by defining a refined approximation of the resolvent
$$\begin{aligned}
\bar G_\eps (z)& =(H-z -i  \sum_{j=1}^n {\eps^j \over j!} B_j )^{-1}  , \\
\bar F_\eps (z) &= |D+i|^{-s} |\eps D+i| ^{ s-n} G_\eps (z) |D+i|^{-s} |\eps D+i| ^{ s-n}\hbox{ with } s >n-1/2\,,
\end{aligned}$$
and  by looking at higher order derivatives of $\bar F_z(\eps)$. We refer to \cite{JMP} for the details of this proof.

\end{proof}

\subsection{Scattering for the evolution problem}
 
A refinement of the  conjugate operator method  allows actually
 to prove that there exists a direction of propagation for the evolution $\exp (-i tH)$. 
We indeed have the following theorem

\begin{Thm}[Mourre]  \label{mourre-thm2}
Let us assume that  $H$ is at least $2-$smooth with respect to $D$, and that the commutator estimate (\ref{commutator-hyp}) holds.
Denote by $P_D^-$ the spectral projectors of $D$ on $\R^-$.
Then, 
$$
\sup_{\lambda \in I }  \| P_D^-  (H-\lambda -  i0)^{-1} |D+i|^{-s} \| \leq C \hbox{ for } s >1\,.
$$
\end{Thm}

\begin{proof}[Sketch of proof]
The arguments here are quite similar to those used in the proof of Theorem \ref{mourre-thm}.
We define  for $ \Im (z) >0$, $\epsilon >0$  and $s>1$
$$ G_z(\eps )= (H-z-  i\eps BB^*)^{-1}, \quad \tilde  F_z(\eps) = P_D^- e^{\eps D}  G_z(\eps) |D+i|^{-s}  \,.$$

\medskip
\noindent
$\bullet$
Let us now look at the derivative of $\tilde F_z(\eps)$ with respect to $\eps$.
$$\frac{d \tilde F_z(\eps) }{ d\eps} = P_D^- e^{\eps D}  \left( DG_z(\eps) -G_z(\eps)P[H,D]PG_z(\eps) \right)  |D+i|^{-s}   \,.$$
The same decomposition as previously shows that
\begin{equation}
\label{diff2}
\begin{aligned}
\frac{d \tilde F_z(\eps) }{ d\eps} &= P_D^- e^{\eps D}  G_z(\eps ) D |D+i|^{-s}  \\
&\quad  +\eps P_D^- e^{\eps D} G_z(\eps)[D, i B^*B]   G_z(\eps) |D+i|^{-s}  \\
& \quad +  P_D^- e^{\eps D} G_z(\eps)(P-I) [D,H]   G_z(\eps) |D+i|^{-s}  \\
&\quad  +  P_D^- e^{\eps D} G_z(\eps) P[D,H] (P-I ) G_z(\eps) |D+i|^{-s}  
\end{aligned}
\end{equation}
Note that the exponential term introduces a shift, so that the derivative $D$ acts only on the right in the first term.
From (\ref{bound1})(\ref{bound2})(\ref{bound3}) and the bound on $\|F_\eps (z)\|$, we have by interpolation that for any $\eta>0$, there exist $C_\eta>0$ and $\eta' >0$ such that
$$ \| G_z(\eps) |D+i|^{-\eta}\| \leq  C_\eta \eps ^{\eta'-1}\,.$$
In particular,  the three first terms in the right hand side of (\ref{diff2}) are integrable.

\medskip
\noindent
$\bullet$
The difficulty here is to get a control on the last term,
and more precisely to prove that $[D,H] (P-I ) G_z(\eps)|D+i|^{-s}  $ is a bounded operator from $L^2$ to $\cH^\eta$ for some $\eta>0$.

\noindent
We start from the following identity
$$\begin{aligned}
 (P-I ) G_z(\eps) |D+i|^{-s} &= (P-I ) G_z(0) \big( 1+i\eps B^*B G_z(\eps) \big)|D+i|^{-s} \\
 &=  (P-I ) G_z(0)  |D+i|^{-s}  + o(\eps^{1/2}) 
 \end{aligned}$$
where the remainder is estimated by  (\ref{bound2}).

\noindent 
We then have to prove that $|D|^\eta [D,H] (P-I ) G_z(0)  |D+i|^{-s} $ is a bounded operator.
The idea is to use the fact that, for any  bounded operator $C$ and  for any $\eta\in [0,1[$
$$\hbox{ if $[C,D]$ is bounded, then }  |D|^\eta  C |D+i|^{-1} \hbox{ is bounded.}$$
The proof of this statement can be found in \cite{mourre2}.
In our case, $C= [D,H] (P-I ) G_z(0)$ is indeed a bounded operator, and the commutator can  be expressed in terms of $[D,H]$ and $[D, [D,H]]$ which are bounded as well.

The last term in the right hand side of (\ref{diff2}) is also integrable.
We then obtain that $\|\tilde F_z(\eps)\|$ has no singularity as $\eps \to 0$.

\end{proof}

\section{Spectral decomposition of the wave operator}

\subsection{Construction of a global escape function}
As explained in the introduction, the existence of a conjugate operator for a pseudo-differential
 operator $H$ of principal symbol $h$ 
is related to the fact that the Hamiltonian dynamics of $h$ admits an escape function.
The construction of the escape function $d$ is therefore the heart of the proof.

\begin{Prop}\label{escape-prop}
 Under the assumptions (M0), (M1) and (M2), there exists an escape function  $d$ for the Hamiltonian $h$
on $\Sigma _{\omega _0}$. 

The function $d$ can then be extended to energies close to $\omega_0$.
\end{Prop}

\noindent
We start by constructing a local change of variables, also called normal form, on each basin of attraction (resp. of repulsion) of the dynamics, in order to have a simple expression for the Hamiltonian $h$.

\begin{Lem}\label{normalform-lem} Let $\gamma$  be a hyperbolic closed
leaf of the foliation ${\cal F}$ with Lyapunov exponent $e^{-2\pi \lambda }\neq 1.$
Denote by  $B_\gamma $  the basin of attraction 
(or ``repulsion'') of $\gamma $.
 There exists a diffeomorphism of 
  $B_\gamma  $ on $(\R/2\pi \Z)_x \times \R_y $ so that the foliation is given   by
$dy +\lambda y dx =0 $ oriented by $dx >0$. 
\end{Lem}

\begin{proof}
Up to a change of orientation, we can assume that $\gamma$ is a stable cycle.

\noindent
$\bullet$ The first step is  to construct the normal form close to the cycle. Let $S$ be a local Poincar\'e section transverse to $\gamma $.
By definition, the Poincar\'e return map $\pi$ sends $S$ on itself. By Sternberg's linearization theorem  for 1D maps \cite{St57},
there is a chart $(I,0)\subset (\R_y,0)$  of $S$ so that 
$$P(y)= \mu y \hbox{  with } 0< \mu =e^{-2\pi \lambda }< 1 .$$
We then choose the   normalisation of the  vector field $V$  tangent to the foliation ${\cal F}$  so that the return time is $2\pi $.
And we denote by $x$ the coordinate starting from $0$ along $S$  and so that $V. {x}=1$.  By construction, $x \in \R/2\pi \Z$.
This construction provides a  foliation ${\cal F}_0$ given by $dy +\lambda  y dx =0 $  and oriented by $dx >0$.

 We have now  two periodic foliations ${\cal F}$ and ${\cal F}_0$: they agree at $x=0$ and $x=2\pi $.
Thus, on each section $x=constant$, there exists  a unique diffeomorphism sending the first onto the second one. 
One  checks that is is smooth. 
In other words, the normal ${\cal F}_0$ is a re-parametrization of  ${\cal F}$, which encodes the geometry of the trajectories.

\bigskip
\noindent
$\bullet$ The second step is then to extend the normal form  globally in the basin. We will use here 
ideas from scattering theory introduced  by Nelson  \cite{Ne69}.

We first  choose the normalization of a generator of the foliation so that it extends smoothly the vector field
 $V $ defined near $\gamma $.  We  denote by $U(t) $  the flow of $V$ on $B$, and by  $U_0(t) $  the flow of $V_0 =\partial _x - \lambda  y \partial _y $ on 
$B_0:=(\R/2\pi \Z )_x \times \R_y $.
Both flows are complete. 

Let us then consider the map
$W:B\rightarrow B_0$ defined
by
$$W (q)= \lim _{t\rightarrow + \infty }    U_0(-t)   U(t)q .$$
The limit clearly exists because $W$ is the identity near $\gamma $ and for any $q\in B$, $U(t) q \rightarrow \gamma $
as ${t\rightarrow + \infty }$.  Both flows are therefore conjugated by $W$. 
\end{proof}

\begin{Rmk} Note that for any $q\in X$ outside the closed leaves, $q$ belongs both to the basin of attraction of a stable cycle, and to the basin of repulsion of an unstable cycle. The foliation ${\cal F}$ at $q$ is therefore conjugated to two  foliations ${\cal F}_0$ with different $\lambda$.
\end{Rmk}

\bigskip
\noindent
As a corollary of the previous Lemma, we obtain the local expression of the Hamiltonian $h$:

\begin{Lem}\label{h0-lem} Let $B_\gamma $ be  the basin of attraction 
(or ``repulsion'') of the hyperbolic closed
leaf $\gamma$ of the foliation ${\cal F}$, and denote by $(x,y,\xi,\eta)$ the coordinates associated to  the normal form introduced in Lemma \ref{normalform-lem}.

Then  there exists a conic neighborhood of the cone $\Gamma_\gamma \subset \Sigma_{\omega_0}$ generated by $B_\gamma$, defined  by  
\[ U_\gamma :=\{ (x,y;\xi,\eta )| ~ |\xi |< c \eta \} \subset T^\star X  \setminus 0  \] 
such that   the 
 Hamiltonian $h$ can be written locally on $U_\gamma$
$$h(x,y,\xi,\eta)-\omega_0=\Phi^2(x,y,\xi, \eta) \left( \frac{\xi}{\eta}-\lambda  y\right)$$
 for some non vanishing function $\Phi$ homogeneous of degree $0$.\end{Lem}

\begin{proof} The orthogonal of the foliation ${\cal F}_0$  is given by 
$\xi -\lambda y\eta =0$. Hence $\eta \ne 0 $ on the cone generated by $B_\gamma$. Without loss of generality, we choose $\eta >0 $ on the cone $\Gamma_\gamma $ generated by 
$\gamma $. 

Any homogeneous function vanishing on $\Gamma_\gamma $ is the product of 
$$h_0 (x,y,\xi,\eta) =\frac{\xi}{\eta }-\lambda y $$
 by a non vanishing multiplier, which is  homogeneous of degree 0 because of the assumption (M0). 
 
 Moreover, in order to get 
 the right stability or unstability according to the sign of $\lambda$, 
this multiplier has  to be non negative. Defining $\Phi$ as its square root gives the expected formula. 
\end{proof}

\bigskip
\noindent

\begin{proof}[Proof of Proposition \ref{escape-prop}] Equipped with these preliminary results, we can   construct the escape function.

\noindent
$\bullet$ The first step is to define  local escape functions on the basin of attraction $B_\gamma$, by using the explicit formula of the Hamiltonian $h$ in the coordinates associated with the normal form. 
For any $k >0$, 
we can define a smooth function $\phi \in C_c^\infty (\R , [0,1])$ 
such that  $\phi \equiv 1 $ on $[-k, k]$  and  satisfying $y\phi'(y) \leq 0 $ on $\R$.
We then set
\[ d_\gamma :=\lambda \eta \phi (y) \,.\]
A straightforward computation shows that
\[  \{  \Phi ^2 h_0 , d_\gamma \}=  \Phi^2 \{  h_0 , d_\gamma \} + h_0 \{  \Phi ^2  , d_\gamma \} = \lambda^2  \Phi^2 \left( \phi -y\phi'    \right) \]
 on the cone  of $ \Gamma_\gamma  \subset \Sigma_{\omega_0} $ generated by $B_\gamma$. 
 
 Hence 
 the function 
$d_\gamma$ 
 satisfies 
$$
\begin{aligned}
\{h,  d_\gamma  \}\geq 0 \hbox{ on } \Sigma_{\omega_0},\\
\{h,  d_\gamma  \}> 0 \hbox{ on } \Gamma_{\gamma,k}
\end{aligned}$$ 
where $\Gamma_{\gamma,k} \subset \Sigma_{\omega_0} $ is the cone generated by  $\{ (x,y)\in B_0 | ~ |y |< k \} $.

\medskip
\noindent
$\bullet$
The global escape function is obtained as a sum of local escape functions. From the assumption (M2), we know that $\Sigma_{\omega_0}$ is a finite union of $B_\gamma$.
Each of these $B_\gamma$ has a normal form, and the corresponding cone $\Gamma_\gamma \subset \Sigma_{\omega_0}$ can be covered by the unions of $\Gamma_{\gamma,k}$ for all $k \in \N^*$
$$ \Sigma_{\omega_0} \subset \bigcup_{\gamma \hbox{ \small{closed leaf}}} \bigcup_{k\in \N^*} \Gamma_{\gamma,k}\,.$$
Since the quotient $Z$ of $\Sigma_{\omega_0}$ by positive dilations is compact, we can extract a finite   covering
$$ \Sigma_{\omega_0} \subset \bigcup_{\gamma \hbox{ \small{closed leaf}}} \Gamma_{\gamma,k_\gamma} \,.$$
By adding the corresponding local escape functions, we then get a function $d$ such that
$\{h,  d  \}>0 \hbox{ on } \Sigma_{\omega_0}\,.$
 \end{proof}

\begin{Rmk}
For any $z\in \Sigma _{\omega _0 } $ which is not in the union of cones generated by
the unstable leaves,  the Hamiltonian trajectory $\phi_t (z)$ converges as $t\rightarrow +\infty $ to the infinity
 of the cones generated by the stable
leaves, which  are invariant conic Lagrangian submanifolds of $T^\star X\setminus 0 $. 

\noindent
The dynamics on these cones is spiraling: using the  coordinates $(x,\eta )$ in these cones, we get
\[ \dot{x}=\Phi^2  \frac{1}{\eta },\quad \dot{y}=-\Phi^2  \frac{\lambda y }{\eta }, ~ \dot{\eta }= \Phi^2 \lambda \]
which gives the spirals 
$$ y = y_0 e^{-\lambda x}, \quad \eta =\eta_0 e^{\lambda x} .$$
\end{Rmk}

\subsection{Construction of the conjugate operator}
We now  construct the conjugate operator $D$ as a pseudo-differential operator of degree $1$ with principal symbol $d$:
 the commutator identity (\ref{commutator-hyp}) will follow
from the estimate on the principal symbol of $B_1 = i[H,D]$ which is
$b_1 = \{ h,d\} $. 

\begin{Thm}
If $d$ is an escape function on $\Sigma _{\omega _0}$, Mourre's theorems \ref{mourre-thm} and \ref{mourre-thm2} apply near $\omega _0$ for any value of $n$.
\end{Thm}

\begin{proof} We first extend $d$ as a smooth function $d_1$ homogeneous of degree $1$ on 
$T^\star X \setminus 0$ which satisfies $\{ h,d_1 \} \geq c >0 $  in some conical neighborhood
$U:=\{ |h-\omega _0| \leq a \} $  of $\Sigma _{\omega _0}$.
And we choose a self-adjoint pseudo-differential operator $D$ of principal symbol $d_1$.

\medskip
We then choose  $\chi,~\bar \chi $ as in Theorem \ref{mourre-thm}, supported in $]\omega_0 -a, \omega _0+a[$ with $\bar \chi(\omega _0)> 0$.
We have then
$$\chi (h) \{  h, d_1 \}  \chi(h )\geq c\bar \chi  (h) .$$
The operators $\chi (H) $ (resp. $\bar \chi(H)$) are pseudo-differential operators of degree $0$ with principal symbols
$\chi (h)$ (resp. $\bar \chi (h)$).
The property (\ref{commutator-hyp}) is  obtained from the sharp 
 G\r{a}rding's inequality (see \cite{Fo89}
pp 129--136):
 \begin{Thm}[Sharp G\r{a}rding] \label{strong-G}
 Let $B$ be a self-adjoint pseudo-differential operator
of degree $0$ on a compact manifold, with  principal symbol $b\geq 0$.
 Then  $$B\geq R $$
  for some operator $R$ of order $-1$.
  \end{Thm}
  It follows then  that
$\chi(H) [H,D] \chi(H) \geq \bar \chi (H)+ R $
where $R$ is a pseudo-differential operator of degree $-1$,  hence compact.

\medskip
The assumptions  on  the iterated brackets follow from the fact that the bracket of
a pseudo-differential operator of degree $0$ with a pseudo-differential operator of degree $1$ is
  a pseudo-differential operator of degree $0$, and hence bounded.  
\end{proof}

\subsection{Characterization of the spectral density}

\begin{Prop}\label{uinfty-prop} Under the assumptions of Theorem \ref{main-thm}, the solution to (\ref{forced}) can be decomposed as
\[ u(t)= 
  e^{-i\omega_0 t} u_\infty + b(t)+\epsilon (t) \]
where
\begin{itemize}
 \item $b(t)$ is a bounded oscillating  function of $t$ with values in $L^2$
whose inverse Fourier transform is supported in
$\sigma _{pp}(H) \cup \sigma _{\rm sing}(H)$, 
\item   $\epsilon (t) $ tends  to $0 $ in $\cH^{-1/2-0}$,
\item and $u_\infty = \left( {H}-\omega_0 - i0 \right)^{-1} f $ is such that $Qu_\infty \in L^2$ for any localization operator whose microlocal support lies in the basin of repulsion of an unstable cycle. 
\end{itemize}
\end{Prop}
Note that this statement is true as soon as we have an escape function for the energy $\omega _0$. 

\noindent
Let us insist on the fact that $u(t)$ stays in $L^2$ for any finite time. Its $L^2$ norm will in general converge to $+\infty$ as $t\to \infty$. 
This asymptotics says  that $u(t)e^{i\omega_0 t}$ converges in $\cH^{-1/2-0}$ as $t\to +\infty$, modulo a  function in $L^\infty_t (L^2)$. 

\begin{proof}
Proposition \ref{uinfty-prop} is now a simple application of Mourre's theory.

\medskip
\noindent
$\bullet$
Theorem \ref{mourre-thm} provides the existence of a closed interval $I$ which contains $0$ such that
\begin{itemize}
\item[(i)] $H-\omega_0$  has no eigenvalue in $I$;
\item[(ii)]   the resolvent
$({H}-\omega_0 - z)^{-1} $ defined for $\Im (z) \ne 0$ admits boundary values at the points  $\omega \in I $, which are H\"older continuous
 in the spaces $O_{1/2+0}$. 
\end{itemize} 
We deduce  that the spectral density of the smooth forcing $f$
  which is given by the Sokhotski-Plemelj formula
 $$ \nu(\lambda) = \frac{1}{ 2i\pi } (({H}-\omega_0- \lambda- i0 )^{-1} f - ({H}-\omega_0- \lambda +  i0 )^{-1} f) \,,$$
 is H\"older continuous with respect to $\lambda$~: $\nu\in C^{0,\mu} ( I, \cH^{-1/2-0})$.
 
In other words,  $\nu $ satisfies the assumptions of the Lemma \ref{lem-int}
with $\cB= O_{-1/2-0}$ and $\cB_0=L^2$. 
The decomposition of the oscillating integral  obtained in Lemma \ref{lem-int}  gives  therefore the expected behavior for $u(t)$.

\medskip
\noindent
$\bullet$
Theorem \ref{mourre-thm2} shows that the resolvent $(H-\omega_0-\lambda - i0)^{-1}$ has no singularity  in the vicinity of unstable cycles.

Let $\gamma$ be an unstable cycle  in $\Sigma_{\omega_0}$, $B_\gamma$ its basin of  repulsion.
Denote by $Q$ a localization operator whose microlocal support lies in a conic neighborhood of the cone
$\Gamma $  generated by $B_\gamma$.
Going back to the construction of the escape function close to $\gamma$, we see that in local coordinates
$$d_{\gamma}( x,y,\xi, \eta) =\lambda  \eta \phi_{\gamma } (y) $$
with $\lambda <0$. In particular $d _{\gamma}$ is strictly negative near $\Gamma $. 
In other words, this means that 
$$  Q = QP_D^-\,.$$ modulo smoothing operators. 

We then conclude that for any $s>1$
$$ \begin{aligned}
Qu_\infty &=   QP_D^- (H-\omega_0-\lambda - i0)^{-1} f \\
& =   QP_D^- (H-\omega_0 -\lambda - i0)^{-1} |D+i|^{-s}\left(  |D+i|^{s}  f\right) \in L^2 (X)\,,
\end{aligned}
$$
which concludes the proof.
\end{proof}

\begin{Rmk}By using a  stronger version of Theorem \ref{mourre-thm2}
 (see \cite{mourre2}  Corollary I.3, Equation (II)), one can  get a  description of the evolution, and not only of the resolvent. In particular, one can prove that
$u(t)$ stays bounded in $L^2$ near the unstable cycles as $t\to +\infty$. 
\end{Rmk}


\section{Precise description of $u_\infty $}\label{wavefront-sec}

Thanks to the normal form defined in Lemma \ref{normalform-lem}, one can actually obtain  a more precise description of $u_\infty$ near any closed leaf $\gamma$. In particular this gives almost explicit formula for the wavefront.
 
\begin{Thm}\label{theo:OIF}
 Under the assumptions (M0), (M1) and (M2), for any $f\in C^\infty(X)$,
the distribution $u_\infty = (H-\omega_0-i0)^{-1} f $ is smooth outside the projections on $X$ of the stable closed leaves.
If $\gamma  $ is such a stable closed leave generating a cone $\Gamma $, then $u_\infty $ is,
microlocally near $\Gamma$, a Fourier integral distribution of order $0$,
whose conic Lagrangian manifold is $\Gamma $ and  whose principal symbol  is a non-homogeneous symbol of order $0$  on $\Gamma $ 
 invariant by the dynamics of $X_h$.

In more explicit  terms, denote by $(x,y,\xi, \eta )$ the coordinates associated with the normal form on the basin $B_\gamma$  
so that  $\gamma  $ projects onto $y=0$.
Let $\chi \in C^\infty_c(\R)$ be a test function which is identically equal to 1 near $0$.
Then the Fourier transform with respect to $y$ of $\chi u_\infty$ satisfies
$$\widehat{\chi u_\infty}(x,\eta)\equiv  \indc_{\eta\geq 0} \sum_{j=0} ^\infty u_{j} (x,\eta )$$ where   $ u_{j}$ is a symbol of degree $-j$.
Furthermore the half density  $u_0\sqrt{dx d\eta }$ is  invariant by the restriction of the Hamiltonian dynamics $X_h$ to $\Gamma $.
\end{Thm}

Once again the challenge here is to transfer informations we have on the classical Hamiltonian dynamics, especially the local representation of $h$ near the closed leaves, to get informations on the solutions to $ (H-\omega_0) v \in C^\infty(X)$. We will then need a counterpart of the normal form at the level of operators.

\subsection{Pseudo-differential normal form}\label{ss:pseudo-normal}

Let $\gamma$ be a cycle on $\Sigma_{\omega_0}$, and $\Gamma$ the cone associated to $\gamma$.
In coordinates associated with the normal form defined in Lemma \ref{normalform-lem},
\[ \Gamma :=\{ (x,y,\xi,\eta) \,/\,  \xi =\lambda y\eta \hbox{ and } \eta >0  \} \subset T^\star \left( \T_x \times \R_y \right)
\setminus 0  \] 
By Lemma \ref{h0-lem},  there exists a conic neighborhood $U$ of the cone  generated by the basin $B_\gamma$   
\begin{equation}
\label{U-def}
 U :=\{ (x,y;\xi,\eta )| ~ |\xi |< c \eta \} \subset T^\star X  \setminus 0  
 \end{equation}
such that   the 
 principal symbol    of $H$ can be written locally on $U$
$$\begin{aligned}
h(x,y,\xi,\eta)&-\omega_0=\Phi^2(x,y,\xi, \eta) h_0(x,y,\xi,\eta) \\
& \hbox{ with } h_0(x,y,\xi,\eta)= \frac{\xi}{\eta}-\lambda  y \hbox{ for some }\lambda \neq 0\,.
\end{aligned}
$$
We further know that the sub-principal symbol of $H$ vanishes.
  
We would therefore  like    to define a reference operator $H_0$ of principal symbol $h_0$ and zero subprincipal symbol. The difficulty is that $h_0 $ is not an admissible  symbol, being not smooth at $\eta =0$.
We therefore choose for $H_0 $ any pseudo-differential operator of degree $0$ 
 whose full symbol is $h_0 $ in the cone $U$ and which is elliptic outside
$U$. Note that the symbol of $H_0$ cannot be  real valued since the sign of $h_0$ changes  on $\Gamma $.  
 
\medskip

We then have the following   pseudo-differential normal form result.

\begin{Prop}\label{pdonormal-prop}
Let $H$ be a   self-adjoint pseudo-differential operator, with principal symbol
$h = h_0 \Phi^2$ in $U$   and vanishing sub-principal symbol. There exists an
elliptic pseudo-differential operator $A$ of principal symbol $1/\Phi$
so that 
\[ A^\star H{ A} -   {H_0}=R   \]
where   $R$  is a smoothing operator when acting on functions which are microlocalized on $U$, i.e.  its full symbol 
is fast decaying in $U$.  
\end{Prop}

\begin{proof} 
We proceed by induction. It is enough to consider symbols in $U$. 
In what follows $\equiv_U  $ means that the difference is smoothing in $U$. 

\noindent

Defining ${A}_0= {\rm Op}^W (1/\Phi)$, we get
first
\[ A_0^\star H {A}_0  \equiv _U  {H_0} + {P}_{2}  \]
for some  pseudo-differential operator  ${P}_{2}$ of order -2. 
This uses the fact that the sub-principal symbol of ${A}_0 ^\star {H} {A}_0$ vanishes; this follows from
the identity which extends the formula for the principal symbol of the commutator:
\[ {\rm sub } (PQ )={\rm sub }(P) \sigma _p (P)+\sigma _p (Q) {\rm sub }(P) +
\frac{1}{2i}\{  \sigma _p (P),\sigma _p (Q)\} \]
and the fact that ${\rm sub }(A_0)=0$ since we have used the Weyl quantization.

We then have   to find, for any integer $n\geq 1$,    some self-adjoint pseudo-differential operators $A_n$
 and  $P_{n+2}$ of respective orders $-n$ and  $-(n+2)$ so that
\[ e^{-iA_n }({H_0} + {P}_{n+1} )e^{iA_n}\equiv_U{H_0}+ P_{n+2 }\]
This gives the following co-homological equation in $U$ for the  principal symbol $a_n$ of $A_n$
\[ \{ h_0 , a_n \} + p_{n+1} =0 \hbox{ in } U\,.\]

In order to solve this equation,
 we decompose $a_n$ and $p_{n+1}$  into Fourier series with respect to  $x$. As we expect $a_n$ to be homogeneous of degree $-n$, we set
\[ a_n=\eta ^{-n} \sum_{k\in \Z}\alpha_{n,k} \left(y, \frac{\xi }{\eta }\right) e^{ikx }, ~p_{n+1}=
\eta ^{-(n+1)} \sum_{k\in \Z}\rho_{n,k} \left(y, \frac{\xi }{\eta }\right) e^{ikx },\]
and we get, using the variable $s=\xi / \eta $:
\[ s \left( \partial _y \alpha_{n,k} + \lambda \partial _s \alpha_{n,k}  \right) +(\lambda n -ik ) \alpha_{n,k} =\rho_{n,k} \]
This singular differential equation admits a unique smooth solution in $\R^2_{y,s}$
 given by
\[ \alpha_{n,k} (y,s)=\frac{1}{\lambda }\int_0^1 \tau ^{n-1-\frac{ik}{\lambda }}\rho_{n,k} \left(y+ \frac{(\tau -1)s}{\lambda },  \tau s    \right) d\tau \]
Furthermore, since the partial derivatives of $\alpha_{n,k}$ are  controlled by 
those of $\rho_{n,k}$ uniformly in $k$, we obtain that the solution $a_n$ has the same regularity as $p_{n+1}$.  
\end{proof}

\subsection{Solutions of  $(H-\omega_0) u \in C^\infty$  near the closed leaves}

The general theory tells us that the wavefront set of any solution to $(H-\omega_0) u \in C^\infty$ is contained in the characteristic manifold
$\Sigma _{\omega _0} $ and is invariant by the dynamics of $X_h$. The idea is then to use the reference operator $H_0$ to get an explicit formula for the singularity. 

\begin{Prop}\label{cor:solutions}
Any distribution $u$ so that  $(H-\omega_0) u$ is smooth,
is given, microlocally near each
closed cycle $\gamma$ in the normal form chart,  by an expansion of the form
\[ u=A \left(  \sum_{k\in \Z} v_k(y+ i0) ^{-1+ ik/ \lambda }e^{ikx} \right) \]
where $A$ is  an elliptic pseudo-differential operator of degree $0$, and the growth of  $(v_k)$ is controlled by (\ref{growth-condition}).

In particular, such a solution is smooth microlocally near a closed cycle $\gamma$
as soon as it is microlocally  $L^2$ near that cycle. 
\end{Prop}

\medskip
\noindent
In order to establish this result, the first step is to solve (microlocally) the equation $H_0 v \in C^\infty$ for the reference operator $H_0$.

 \begin{Lem}\label{prop:sol}
Any  solution of $H_0 v \in C^\infty$, whose wavefront set  is contained in  the set $U$ defined by (\ref{U-def}),
  admits the following expansion, up to smooth terms,
\[ v (x,y) \equiv  \sum_{k\in \Z} v_k (y + i0) ^{-1+ ik/ \lambda }e^{ikx}\,.\] 
\end{Lem}

\begin{proof}
We know already that the wave front set of $v $ is included in the characteristic manifold
$\xi = \lambda y\eta $ intersected with $U$. Composing on the left by $D:= \partial _y$, we get that 
$v$ satisfies
$(\partial _x -\lambda (y\partial _y + 1) )v \in C^\infty $. 

Using a  Fourier decomposition in  $x$,  we are then reduced to solve
$$\left( i k  - \lambda (y\partial _y + 1)\right)v_k  = -\lambda g_k  $$
for some smooth $g_k$.
This is a simple linear ordinary differential equation  whose solution  is given, for $y>0 $,   by
\[ v_k(y)= \int_0^1 s^{-ik/\lambda} g_k (ys) ds + l_k^+ y_+   ^{-1+ik/\lambda}          \] 
A similar expression holds for $y<0$.  The first term is smooth if $g_k$ is.

Summing the Fourier expansions and using again the Sokhotski-Plemelj theorem, we get
\[ v  \equiv \sum_k  v_k^+  (y+i0) ^{-1+ ik/ \lambda }e^{ikx} + \sum_k v_k^- (y-i0) ^{-1+ ik/ \lambda }e^{ikx}\]
Now, since the Fourier transform of the second sum is supported by $\eta <0$, it should vanish because of the 
wavefront set assumption. 
\end{proof}

\medskip
\noindent
We then  need  to characterize the admissible sequences $(v_k)$.

\begin{lem}\label{growth-lem} 
The sum 
$T:= \sum v_k (y+i0)^{-1+ik/\lambda }e^{ikx} $
defines a tempered distribution if and only if
\begin{equation}
\label{growth-condition}
v_k   e^{\pi (k/\lambda)_-} \hbox{  is of polynomial growth.}
\end{equation}
Furthermore,  if   $T\in L^2_{\rm loc}$ near $y=0$, the sequence $(v_k)$ is identically zero.
\end{lem}

\begin{proof}
We recall that the  Fourier transform extends to tempered distributions (see \cite{GS64}, section 2.3). In particular we have,
 for $\alpha \in \R\setminus 0 $:
\[   {\cal F} \left((y+i0)^{-1+i\alpha}\right)(\eta)=- \frac{2\pi i e^{-\alpha \pi/2}}{\Gamma (1-i\alpha)}\eta _+^{-i\alpha } \equiv \gamma_\alpha  \eta _+^{-i\alpha }\]
We know also that
\[ |\Gamma (1-i\alpha )|= \sqrt{\frac{\pi \alpha }{\sinh \pi \alpha}}\] 
so that
\[ |\gamma _\alpha| = 2\pi  e^{-\alpha \pi/2}\sqrt{\frac{\sinh \pi \alpha}{\pi \alpha}} \sim _{\alpha \to \infty}Ê\sqrt{ \frac{2\pi}{|\alpha |}}e^{\pi \alpha_-}\,. \]
 
 The condition for $T$ to be a tempered distribution is that the series of its  Fourier coefficients grows at most polynomially, which is exactly condition (\ref{growth-condition}).
 
 Since none of the elementary functions $(y+i0)^{-1+i\alpha}$ is in $L^2_{loc}$, the only way that $T\in L^2_{\rm loc}$ near $y=0$ is that all Fourier coefficients vanish.
\end{proof}

\bigskip
Equipped with this characterization of the solutions to $H_0 v \in C^\infty$, we can now deduce the structure of  singularities in $u_\infty$ using  the normal form.

\begin{proof}[Proof of Proposition \ref{cor:solutions}]
Let $u$ be a solution to $(H-\omega_0) u \in C^\infty$.
We know that 
$WF(u)\subset \Sigma_{\omega_0}$.

\medskip
\noindent
Using a microlocal partition of unity on $\Sigma_{\omega_0}$, we can decompose 
$$ u =\sum_{\gamma} u_\gamma $$  
where $u_\gamma$ is a solution to  $(H-\omega_0) u _\gamma \in C^\infty$, microlocalised on the cone $\Gamma_\gamma$ generated by the cycle $\gamma$.
We have then $WF(u_\gamma )\subset \Gamma_\gamma$.

\medskip
\noindent
By Proposition \ref{pdonormal-prop}, there exists an operator $A$  elliptic on a conic neighborhood of $\Gamma_\gamma$ such that the equation $Hu_\gamma \in C^\infty $ rewrites
$$H_0A^{-1} u_\gamma  \in C^\infty, \quad WF( A^{-1} u_\gamma) \subset \Gamma_\gamma $$
We get then, from Lemma \ref{prop:sol} and \ref{growth-lem}
$$A^{-1} u\equiv  \sum_k  v_k  (y+i0) ^{-1+ ik/ \lambda }e^{ikx}$$. 
\end{proof}

\subsection{Proof of Theorem \ref{theo:OIF}}

\noindent
$\bullet$ By Proposition \ref{uinfty-prop}, we know   that $u_\infty = (H-\omega_0-i0)^{-1} f$  is $L^2$ near the unstable cycles.
 Then, by Proposition \ref{cor:solutions},
we deduce  that $u_\infty $ is smooth near the unstable cycles.

\begin{Rmk} Note  that the same argument shows that any $L^2$ eigenfunction of $H$  is actually smooth.
\end{Rmk}

\noindent
$\bullet$ Since the wavefront set of $u_\infty$  is invariant by the dynamics, and that $u_\infty $ is smooth near the unstable cycles, 
we further obtain  that $u_\infty $ is smooth in the unstable basins.
In view of the general form of the solutions of $(H-\omega _0) u \in C^\infty$ on the cone generated by any basin $B_\gamma$, this regularity implies that 
$(v_k)$ is fast decaying.

\noindent
$\bullet$ Any  distribution  microlocalized on $U$ of the
form
$$v(x,y)=\sum _k v_k (y+i0)^{-1+ik/\lambda } e^{ikx },$$
with   fast decaying coefficients $(v_ke^{\pi (k/\lambda)_-})$ is 
 a Fourier integral distribution associated to $\Gamma $.
This follows easily by taking the $y-$Fourier transform
$$\hat{v}(x,\eta)=\sum \gamma _{k/\lambda } v_k  \eta_+^{-ik/\lambda }e^{ikx} $$
which is a symbol of degree $0$  in $\eta $. 
Then  applying the operator $A$ keeps that property.

\begin{Rmk}
Note that, once we know that the wavefront set is located on the cones
 generated by stable  orbits, this singular behavior could be also obtained by semiclassical arguments, or by boundary layer techniques.
\end{Rmk}


\section{Conclusion}

We have shown that the phenomenon of concentration of the energy on attractors is a very general feature of forced dynamics 
governed by a  pseudo-differential operator with principal symbol  homogeneous of degree 0, associated to a non singular foliation. 

\bigskip
In the recent paper \cite{CdV}, the first author has extended this work to more singular situations. 
There are still  the assumptions that $\Sigma_{\omega _0} $ is nondegenerate, and  that it satisfies a Morse-Smale property which
 is crucial to build escape functions and apply Mourre theory.
But  the assumption that the projection $\pi $ is a finite covering  is removed by using suitable canonical transformations
and their quantizations as in \cite{weinstein,weinstein2}. We can then admit singular foliations with singular
 hyperbolic points which are foci, nodes or saddles.
Actually, these singular points and the corresponding generic normal forms are already studied in the context of implicit differential 
equations (see \cite{Ar88}).

\bigskip
However, even extended to singular foliations, this general theory fails to apply to the specific situation of internal waves observed in lab experiments:
\begin{itemize}
\item we would indeed need to understand how to catch the effects of boundaries which create discontinuities in the frequency space;
\item furthermore it would be necessary to add viscous effects, which are no more negligible when the wavelengths become large
 (see \cite{rieutord, ogilvie}).
\end{itemize}
These questions are major challenges for the mathematical analysis.

\bigskip
Another interesting perspective comes from the following remark. Although it is linear, the system we have studied here exhibit
 the main features of wave turbulence.
We have indeed proved that
\begin{itemize}
\item with a stationary forcing at wavelengths $O(1)$, small scales will be excited;
\item the energy cascade is given by  the frequency distribution  in the wavefront set $WF(u_\infty)$.
\end{itemize}
These properties are clearly related to the quasi-resonant mechanism, or in other words to  the fact that the
 spectrum of the operator is continuous. This could be an indication for the study of turbulence due to some weak coupling of waves.

\bigskip

\appendix

\section{Appendix: the integrable case}

Let us consider the operator $H$ on $L^2(X,|dx|)$, where $X$ is the torus $(\R/2\pi \Z)^d $,
defined by
\[\forall  n \ne 0,~ H e_n =h(n )e_n \]
and $H e_0=0$, 
with 
$e_n (x)={\rm exp}(i\langle n | x \rangle )$ for $n \in \Z^d$
and $h:\R^d \setminus 0 \to  \R $ is smooth and homogeneous of degree $0$.
The spectrum of $H$ is the union of $\{0\} $ and of the interval $J:=[ \min h ,\max h  ]$,  and is pure point dense
with eigenvalues  $\big(h(n )\big)_{n\in \Z^d}$.

\medskip
The solution to the forced wave equation 
is then given by
\begin{equation}
\label{integrable}
 u(t) =\sum_{h(n )\ne 0,~n \in \Z^d\setminus 0} a_n \frac{1-e^{-it h(n)}}{h(n )}e_n +
it \sum _{h(n )= 0,~n \in \Z^d\setminus 0}a_n    e_n
\end{equation}
where the $a_n $ are the Fourier coefficients of the forcing $f$
(assuming without loss of generality that $a_0=0$).


\subsection{The resonant case}

Consider first the case when $0$ is an eigenvalue of $H$, i.e. there exists $n_0 \ne 0 $ with $h(n_0 )=0$.
Let us assume in addition  that te non degeneracy condition is satisfied:  $h'(n_0) \neq 0$.

Using the euclidean division of $n$ by $n_0$, we get that for $n$ large enough 
$$h(n) \neq 0 \quad \Rightarrow \quad |h(n)|\geq \frac12 |h'(n_0)| { n_0 \over n }  \,.$$
Hence the first part of the sum in (\ref{integrable}) is bounded in $L^2$ as soon as 
$$\sum |  a_n|^2 |n |^2 <\infty ,$$
 i.e. as soon as the forcing is in the Sobolev space $H^1 (X)$.

In this case
\[ u(t) =it f_0 + O_{L^2}(1) \]
where $f_0$ is the orthogonal projection of $f$ onto $\Ker H$. 

\subsection{The non resonant case}

We now consider the non resonant case, with the diophantine condition
\[\forall n \in Z^d \setminus 0,~  |h(n)|\geq C|n |^{-\alpha } \]
with $C>0,~\alpha >0$. 
Then, if $f$ is smooth enough (depending on $\alpha$), the solution to the forced wave equation given by (\ref{integrable}) remains bounded in $L^2$.

\medskip
In order to compare this result to the non integrable case, we will translate the diophantine condition on the symbol  $h$ into a dynamical condition.
Let  us look at the  classical dynamics on the set $h(p)=0$. 
The momentum is constant, so that we can look at the dynamics on a torus $X\times \{ p_0 \}$ with $h(p_0)=0$.
Assume for simplicity that $h$ vanishes only on the line generated by $p_0$ with $|p_0|=1$, and that we have the non degeneracy condition $h'(p_0) \neq 0$. 
The dynamics is linear, given by the vector field
\[ Y=\sum_{j=1}^d  \frac{\d h}{\d p_j}(p_0)\d _{x_j} \]  
For any $n$, define $n_1 = n/|n|$. We have, using $Y.p_0=0$, 
\[ h(n_1) -h(p_0)= Y.n_1 + O(|n_1 -p_0|^2)\]
\begin{itemize}
\item If $|n_1 -p_0|\leq c/|n |^{\alpha /2} $, we get
$|Y.n_1 | \geq C_1/|n |^{\alpha } $ with $C_1>0$.
\item 
If $|n_1 -p_0|\geq c/|n |^{\alpha /2} $,
we get  $| h(n_1) -h(p_0)|\geq C_2 /|n |^{\alpha /2}$, and
hence  $|Y.n_1 | \geq C_3/|n |^{\alpha /2}$.
\end{itemize}
Finally we get that $Y$ satisfies a Diophantine condition. 
A similar reasoning shows that the converse is true.

\begin{Prop}
The following diophantine conditions are equivalent:
$$\exists C,\alpha >0 ,\quad 
\forall n \in \Z^d\setminus 0,~ |h(n) |\geq C|n |^{-\alpha }\,,$$
$$\exists D,\beta >0, \quad 
\forall n \in \Z^d\setminus 0,~  |Y.n | \geq D|n |^{-\beta }\,.$$
\end{Prop}

\noindent
{\bf Acknowledgements.} We would like to thank E. Ghys  for enlightening discussions on escape functions. We  also thank J. Sj\"ostrand,
 S. Dyatlov and M. Zworski
for their very useful comments on the first version of this paper.
Finally we express our gratitude to T. Dauxois and D. Bresch for their careful reading and their suggestions to make the paper accessible to a broader audience.


\bibliographystyle{plain}

\end{document}